\documentclass[a4paper,UKenglish,cleveref, autoref, numberwithinsect, thm-restate]{lipics-v2021}

\pdfoutput=1 
\hideLIPIcs  

\nolinenumbers
\usepackage{xcolor}

\title{Subgraph Counting in Subquadratic Time for Bounded Degeneracy Graphs} 

\author{Daniel Paul-Pena}{University of California, Santa Cruz, United States }{dpaulpen@ucsc.edu}{https://orcid.org/0009-0008-1073-6173}{}

\author{C. Seshadhri}{University of California, Santa Cruz, United States}{sesh@ucsc.edu}{https://orcid.org/0000-0003-2163-3555}{}

\authorrunning{D. Paul-Pena and C. Seshadhri} 

\Copyright{Daniel Paul-Pena and C. Seshadhri} 



\category{} 

\relatedversion{} 

\funding{Both authors are supported by NSF CCF-1740850, CCF-1839317, CCF-2402572, and DMS-2023495.
}
\EventEditors{John Q. Open and Joan R. Access}
\EventNoEds{2}
\EventLongTitle{42nd Conference on Very Important Topics (CVIT 2016)}
\EventShortTitle{CVIT 2016}
\EventAcronym{CVIT}
\EventYear{2016}
\EventDate{December 24--27, 2016}
\EventLocation{Little Whinging, United Kingdom}
\EventLogo{}
\SeriesVolume{42}
\ArticleNo{23}

\newcommand{\maybefootnote}[2]{%
	\ifcsname firsttime#1\endcsname
	\else
	\expandafter\gdef\csname firsttime#1\endcsname{}%
	\footnote{#2}%
	\fi
}

\mathchardef\mhyphen="2D
\newcommand{\reducible}[1]{${#1}$-reducible}
\newcommand{\computable}[1]{${#1}$-computable}
\newcommand{\reduced}[1]{G_{#1}}

\newcommand{\cycle}[1]{\cC_{#1}}

\newcommand{\simplex}[1]{\cS_{#1}}
\newcommand{\hyperone}{\cH_1}
\newcommand{\hypertwo}{\cH_2}
\newcommand{\hyperthree}{\cH_\triangle}
\newcommand{\diamondgraph}{\cD}

\newcommand{\colSetSub}{\text{Col-}\cS}

\newcommand{\maxoutdeg}{d}

\newcommand{\expandG}{G'}
\newcommand{\expandGOdd}{G''}

\newcommand{\extension}[3]{ext\left(#1,#2;#3\right)}

\newcommand{\WSub}[2]{\mathrm{\text{Col-WSub}}_{#2}(#1)}
\newcommand{\WSubNI}[1]{\mathrm{\text{Col-WSub}}_{#1}}
\newcommand{\Hom}[2]{\mathrm{Hom}_{#2}(#1)}
\newcommand{\Sub}[2]{\mathrm{Sub}_{#2}(#1)}
\newcommand{\IndSub}[2]{\mathrm{IndSub}_{#2}(#1)}

\newcommand{\LICL}{LICL}

\newcommand{\Reachable}{Reach}

\newcommand{\Spasm}{Spasm}

\newcommand{\degen}{\kappa}

\newcommand{\dtw}{\tau}

\newcommand{\dagtree}{DAG-tree decomposition}
\newcommand{\dagtreewidth}{DAG-treewidth}

\newcommand{\ignore}[1]{}

\newcommand{\cB}{\mathcal{B}}
\newcommand{\cC}{{\cal C}}
\newcommand{\cD}{\mathcal{D}}
\newcommand{\cE}{{\cal E}}

\newcommand{\cH}{{\cal H}}

\newcommand{\cP}{\mathcal{P}}

\newcommand{\cS}{\mathcal{S}}
\newcommand{\cT}{{\cal T}}

\newcommand{\eps}{\varepsilon}

\newcommand{\NN}{\mathbb{N}}

\newcommand{\Sec}[1]{\S \ref{sec:#1}} 
\newcommand{\Eqn}[1]{\hyperref[eq:#1]{(\ref*{eq:#1})}} 
\newcommand{\Fig}[1]{{Fig.\,\ref{fig:#1}}} 
\newcommand{\Tab}[1]{\hyperref[tab:#1]{Tab.\,\ref*{tab:#1}}} 
\newcommand{\Table}[1]{\hyperref[tab:#1]{Table\,\ref*{tab:#1}}} 
\newcommand{\Thm}[1]{\hyperref[thm:#1]{Theorem\,\ref*{thm:#1}}} 
\newcommand{\Fact}[1]{\hyperref[fact:#1]{Fact\,\ref*{fact:#1}}} 
\newcommand{\Lem}[1]{\hyperref[lem:#1]{Lemma\,\ref*{lem:#1}}} 
\newcommand{\Prop}[1]{\hyperref[prop:#1]{Prop.~\ref*{prop:#1}}} 
\newcommand{\Cor}[1]{\hyperref[cor:#1]{Corollary~\ref*{cor:#1}}} 
\newcommand{\Conj}[1]{\hyperref[conj:#1]{Conjecture~\ref*{conj:#1}}} 
\newcommand{\Def}[1]{\hyperref[def:#1]{Definition~\ref*{def:#1}}} 
\newcommand{\Alg}[1]{\hyperref[alg:#1]{Alg.~\ref*{alg:#1}}} 
\newcommand{\Clm}[1]{\hyperref[clm:#1]{Claim~\ref*{clm:#1}}} 
\newcommand{\Obs}[1]{\hyperref[obs:#1]{Observation~\ref*{obs:#1}}} 
\newcommand{\Rem}[1]{\hyperref[rem:#1]{Remark~\ref*{rem:#1}}} 
\newcommand{\Con}[1]{\hyperref[con:#1]{Construction~\ref*{con:#1}}} 
\newcommand{\Step}[1]{\hyperref[step:#1]{Step~\ref*{step:#1}}} 
\newcommand{\Assumption}[1]{\hyperref[assm:#1]{Assumption\,\ref*{assm:#1}}} 

\usepackage{algorithm}
\usepackage{algpseudocode}
\usepackage{amsmath}
\usepackage{comment}
\usepackage{tikz}
\usepackage{paralist}
\usepackage{caption}
\usetikzlibrary{positioning}
\usepackage{mathtools,amsfonts}

\usepackage{thm-restate}
\usepackage{standalone}
\usepackage{standalone}

\begin{document}
	
	\maketitle
	
	\begin{abstract}
		We study the classic problem of subgraph counting, where we wish to determine the number of occurrences of a fixed pattern graph $H$ in an input graph $G$ of $n$ vertices. Our focus is on \emph{bounded degeneracy} inputs, a rich family of graph classes that also characterizes real-world massive networks. 
		Building on the seminal techniques introduced by Chiba-Nishizeki (SICOMP 1985), a recent line of work has built subgraph counting algorithms for bounded degeneracy graphs. Assuming fine-grained complexity conjectures, there is a complete characterization of patterns $H$ for which linear time subgraph counting is possible. For every $r \geq 6$, there exists an $H$ with $r$ vertices that cannot be counted in linear time.
		
		In this paper, we initiate a study of subquadratic algorithms for subgraph counting on bounded degeneracy graphs. We prove that when $H$ has at most $9$ vertices, subgraph counting can be done in $\tilde{O}(n^{5/3})$ time. As a secondary result, we give improved algorithms for counting cycles of length at most $10$. Previously, no subquadratic algorithms were known for the above problems on bounded degeneracy graphs. 
		
		Our main conceptual contribution is a framework that reduces subgraph counting in bounded degeneracy graphs to counting smaller hypergraphs in arbitrary graphs. We believe that our results will help build a general theory of subgraph counting for bounded degeneracy graphs.
	\end{abstract}

	\section{Introduction} \label{sec:intro}
	
	The fundamental algorithmic problem of subgraph counting in a large input graph has a long
	and rich history~\cite{Lo67, DiSeTh02, FlGr04, DaJo04, Lo12, AhNeRo+15, CuDeMa17, PiSeVi17, SeTi19}.
	There are applications in logic, properties of graph products, partition functions in statistical physics, database theory, machine learning, and network science~\cite{ChMe77,BrWi99,DrRi10,BoChLo+06,PiSeVi17,DeRoWe19,PaSe20}. 
	We are given a \emph{pattern} graph $H = (V(H), E(H))$, and an \emph{input} graph $G = (V(G), E(G))$.
	All graphs are assumed to be simple. We use $\Sub{G}{H}$ to denote the problem of computing the number of (not necessarily induced) subgraphs of $H$ in $G$, that is, the number of subgraphs of $G$ isomorphic to $H$. 
	
	If the pattern is part of the input, this problem becomes $\mathbb{NP}$-hard,
	as it subsumes subgraph isomorphism. Often, one thinks of the pattern $H$ as fixed, and running times are parameterized in terms of the properties (like the size) of $H$.
	Let us set $n = |V(G)|$ and $k = |V(H)|$. 
	There is a trivial brute-force $O(n^k)$ algorithm, but we do not expect
	$O(n^{k-\eps})$ algorithms for general $H$ for some constant $\eps> 0$~\cite{DaJo04}. 
	
	The rich field of subgraph counting focuses on restrictions on the pattern or the input,
	under which non-trivial algorithms and running times are possible~\cite{ItRo78,AlYuZw97,BrWi99,DrGr00,DiSeTh02,DaJo04,BoChLo+06,CuDeMa17,Br19,RoWe20}. 
	Given the practical importance of subgraph counting, there is a special focus on linear time
	(or small polynomial) running times.
	
	Inspired by seminal work of Chiba-Nishizeki~\cite{ChNi85}, a recent line of work has focused
	on building a theory of subgraph counting for \emph{bounded degeneracy graphs}~\cite{Br19,BePaSe20,BePaSe21,BrRo22,BeGiLe+22}.
	These are classes of graphs where all subgraphs have a constant average degree. 
    This work culminated in results of Bera-Pashanasangi-Seshadhri and Bera-Gishboliner-Levanzov-Seshadhri-Shapira~\cite{BePaSe21, BeGiLe+22}.
	We now have precise dichotomy theorems characterizing linear time subgraph counting in bounded degeneracy graph.
    When $H$ has at most $5$ vertices, then $\Sub{G}{H}$ can
	be determined in linear time (if $G$ has bounded degeneracy). For all $k \geq 6$,
    there is a pattern $H$ on $k$ vertices that cannot be counted in linear time, assuming
    fine-grained complexity conjectures.
	The following question is the next step from this line of work.
	\smallskip
	
	\emph{When can we get subquadratic algorithms for subgraph counting (when $G$ has bounded degeneracy)?
		Are there non-trivial algorithms that work for all $H$ with $6$ (or more) vertices?
	}
	
	\smallskip
	Before stating our main results, we offer some justification for this problem.
	
	{\bf Bounded degeneracy graphs:} This is an extremely rich family of graph classes, containing
	all minor-closed families, bounded expansion families, and preferential attachment graphs~\cite{Se23}.
	Most massive real-world graphs, like social networks, the Internet, communication networks, etc.,
	have low degeneracy (\cite{GoGu06,JaSe17,ShElFa18,BeChGh20,BeSe20}, also Table 2 in~\cite{BeChGh20}). 
	The degeneracy has a special significance in the analysis of real-world graphs, since it is intimately tied to the technique of ``core decompositions''~\cite{Se23}. 
	Most of the state-of-the-art practical subgraph counting algorithms use algorithmic techniques
	for bounded degeneracy graphs~\cite{AhNeRo+15,JhSePi15,PiSeVi17,OrBr17,JaSe17,PaSe20}.
	
	{\bf Subquadratic time:} From a theoretical perspective, the orientation techniques of
	Chiba-Nishizeki and further results~\cite{ChNi85,BePaSe20,BePaSe21,BeGiLe+22} are designed for linear time algorithms.
	The primary technical tool is the use of DAG-Tree decompositions, introduced
	in a landmark result of Bressan \cite{Br19,Br21}. Bressan's algorithm yields running times
	of the form $n^r$, where $r \in \NN$ and is the \emph{DAG-treewidth} of $H$.
	It is known that linear time algorithms are possible iff the DAG-treewidth is one \cite{BeGiLe+22}.
	It is natural to ask if the DAG-treewidth being two is a natural barrier. 
	\emph{Subquadratic} algorithms would necessarily require a new technique, other than
	DAG-treewidth. It would also show situations where the DAG-treewidth can be beaten.
	
	From a practical standpoint, the best exact subgraph counting codes (the ESCAPE package~\cite{escape}) use methods similar to the above results. Algorithms for bounded degeneracy graphs have been remarkably successful in dealing with modern large networks. Subquadratic algorithms could provide new practical tools for subgraph counting.
	
	{\bf The focus on $6$ vertices and the importance of cycle patterns:} The problem
	of counting all patterns of a fixed size is called \emph{graphlet analysis}
	in bioinformatics and machine learning \cite{Pr07,ShViPe+09}. Current scalable exact counting codes go to
	5 vertex patterns~\cite{escape}, which is precisely the theoretical barrier seen in \cite{BePaSe20}.
	Part of our motivation is to understand the complexity of counting all $6$ vertex patterns
	(in bounded degeneracy graphs). 
	
	The seminal cycle detection work of Alon-Yuster-Zwick also gave algorithms 
	parameterized by the graph degeneracy~\cite{AlYuZw97}. But most of their results are for cycle \emph{detection}, whereas counting is arguably the more relevant problem. It is natural to ask if counting is also feasible with similar running times.
	
	\subsection{Main Results} \label{sec:results}
	
	Our main result shows that patterns with at most $9$ vertices can be counted in subquadratic time. Let $n$ be the number of vertices and $\degen$ be the degeneracy of the input graph $G$.
    The degeneracy $\kappa$ is the maximum value, over all subgraphs of $G$, of the minimum degree of the subgraph.
    In what follows, $f:\NN \to \NN$ denote some explicit function.

	\begin{restatable}{theorem}{main}(Main Theorem)\label{thm:main} 
		There is an algorithm that computes\maybefootnote{1}{We can obtain a similar result for the problem of counting only induced subgraphs $\IndSub{G}{H}$, as it can be expressed as a linear combination of $\Sub{G}{H'}$ for some patterns $H'$ with $V(H')=V(H)$.} $\Sub{G}{H}$ for all patterns $H$ with at most $9$ vertices
		in time $f(\degen) \tilde{O}(n^{5/3})$.\maybefootnote{2}{We express our results parameterizing by the degeneracy of the input graph $G$. Note that if $G$ is from a class with bounded degeneracy, then $f(\degen)$ is constant and we can ignore that term.}
	\end{restatable}

	Recall that previous works gave (near) linear time algorithms when $H$ had at most $5$ vertices \cite{BePaSe20}. The best subgraph counting algorithm for bounded degeneracy graphs is Bressan's algorithm, which runs in at least quadratic (if not cubic) time for many patterns with $6$ to $9$ vertices.
	
    Additionally, we construct an explicit $10$-vertex pattern that we conjecture cannot be counted in subquadratic time in the bounded degeneracy setting. We are able to relate the complexity of counting that pattern to counting a specific hypergraph in general graphs, which we believe can not be done in subquadratic time. (More in \Sec{invert} and \Sec{hardness}.)
	
	\subsubsection{Counting cycles}

	As a secondary result, we are able to show that cycle counting in bounded degeneracy graphs can be done even faster. These results are tight, in the sense that any improvement will imply beating long
    standing cycle detection algorithms for sparse graphs. Let $\cycle{k}$ denote the $k$-cycle and $d_k$ be the exponent (in terms of edges) of the fastest algorithm for $k$-cycle detection. 
    Gishboliner-Levanzov-Shapira-Yuster (henceforth GLSY) recently showed that \emph{homomorphism} counting of $\cycle{2k}$ in bounded degeneracy graphs can be done in time $O(n^{d_k})$~\cite{GiLeSh+23}. 
    We prove that this complexity can be matched for the problem of \emph{subgraph} counting. Cycle counting is more challenging, since it involves compute linear combinations
    of homomorphisms of non-cyclic patterns (which requires other techniques).
	
	\begin{theorem} \label{thm:cycles}

		\begin{asparaitem}
			\item There is an algorithm that computes $\Sub{G}{\cycle{6}}$ and $\Sub{G}{\cycle{7}}$ in time $f(\degen)\tilde{O}(n^{d_3}) \approx f(\degen)\tilde{O}(n^{1.41})$. Additionally, no $f(\degen)o(n^{d_3})$ algorithm exists unless there exists a $o(m^{d_3})$ algorithm for counting triangles.
			\item There is an algorithm that computes $\Sub{G}{\cycle{8}}$ and $\Sub{G}{\cycle{9}}$ in time $f(\degen)\tilde{O}(n^{d_4}) \approx f(\degen)\tilde{O}(n^{1.48})$. Additionally, no $f(\degen)O(n^{d_4})$ algorithm exists unless there exists a $o(m^{d_4})$ algorithm for counting $4$-cycles.
			\item There is an algorithm that computes $\Sub{G}{\cycle{10}}$ in time $f(\degen)\tilde{O}(n^{d_5}) \approx f(\degen)\tilde{O}(n^{1.63})$. Additionally, no $f(\degen)o(n^{d_5})$ algorithm exists unless there exists a $o(m^{d_5})$ algorithm for counting $5$-cycles.
		\end{asparaitem}
	\end{theorem}

	Previously, for bounded degeneracy graphs, subquadratic results were only known
    for detecting cycles, by a result of Alon, Yuster and Zwick \cite{AlYuZw97}. \Thm{cycles} improves or matches their bounds in all cases, despite solving the harder problem of counting. 
    The lower bounds relating to cycle counting in general graphs follow directly from the techniques of GLSY~\cite{GiLeSh+23}.
    We note that the exponents $d_k$ have not been improved for twenty years~\cite{AlYuZw97, YuZw04}.
    As a side corollary of our methods, we also get a better algorithm for counting $5$-cycles in arbitrary graphs (\Cor{5-cycle}).

	\subsection{Main Ideas} \label{sec:ideas}
	
	The theorems above are obtained from a new reduction technique that converts
	homomorphism counting in bounded degeneracy graphs to subgraph counting in 
	arbitrary graphs for some specific patterns.

The starting point for most subgraph counting algorithms for bounded degeneracy graphs
	is to use \emph{graph orientations}~\cite{Se23}. A graph $G$ has bounded degeneracy iff
	there exists an acyclic orientation $\vec{G}$ such that all vertices have bounded \emph{outdegree}.
	(An acyclic orientation is obtained by directing the edges of $G$ into a DAG.)
	Moreover, this orientation can be found in linear time~\cite{MaBe83}. To count $H$-subgraphs
	in $G$, we consider all possible orientations $\vec{H}$ of $H$ and compute
	(the sum of) all $\Sub{\vec{G}}{\vec{H}}$.
	
	The approach formalized by Bressan~\cite{Br21} and Bera-Pashanasangi-Seshadhri~\cite{BePaSe20} is to break $\vec{H}$ into a collection of (out)directed trees rooted at the sources of $\vec{H}$. 
	The copies of each tree in $\vec{G}$ can be enumerated in linear time, since outdegrees
	are bounded. We need to figure out how to "assemble" these trees into copies of $\vec{H}$.
	
	Bressan's DAG-tree decomposition gives a systematic method to perform this assembly,
	and the running time is $O(n^\tau)$, where $\tau$ is the ``DAG-treewidth'' of $\vec{H}$.
	The definition is technical, so we do not give details here. 
	Also, this method only gives the homomorphism count (edge-preserving maps from $H$ to $G$), and we require further
	techniques to get subgraph counts~\cite{CuDeMa17}. 
	The main contribution of the linear time dichotomy theorems is to completely characterize patterns $H$ such that all orientations $\vec{H}$ have DAG-treewidth one~\cite{BePaSe21,BeGiLe+22}. 
	
	Since $\tau$ is a natural number, to get subquadratic time algorithms, we need new ideas.
	
	\subsubsection{The 6-cycle} It is well-known (from~\cite{AlYuZw97}) that linear-time orientation based methods hit an obstruction at $6$-cycles. 
    Consider the oriented $6$-cycle in the left of \Fig{cycles}. It is basically a ``triangle'' of out-out wedges (as given by the red lines); indeed, one can show that $6$-cycle counting in bounded degeneracy graphs is essentially equivalent to triangle counting in arbitrary graphs~\cite{BePaSe20,BePaSe21,BeGiLe+22}.
    This observation can be converted into an algorithm, as was shown by GLSY~\cite{GiLeSh+23}. Starting with $\vec{G}$, we create a new graph with edges (the red lines) corresponding to the endpoints of out-out wedges. Since $\vec{G}$ has bounded outdegree, the number of edges in the new graph is $O(n)$. Every triangle in the new graph corresponds to a (directed) $6$-cycle homomorphism in $\vec{G}$.
	Triangle counting in the new graph can be done in $O(n^{1.41})$ time (or $n^{3/2}$ time using a combinatorial algorithm)~\cite{AlYuZw97}. 
    Getting the subgraph count is more involved, but we can use existing methods that reduce to homomorphism counting~\cite{CuDeMa17}.
	
	\begin{figure}[t]
		\centering
		\begin{minipage}{.47\linewidth}
			\centering
			\includegraphics[width=\textwidth]{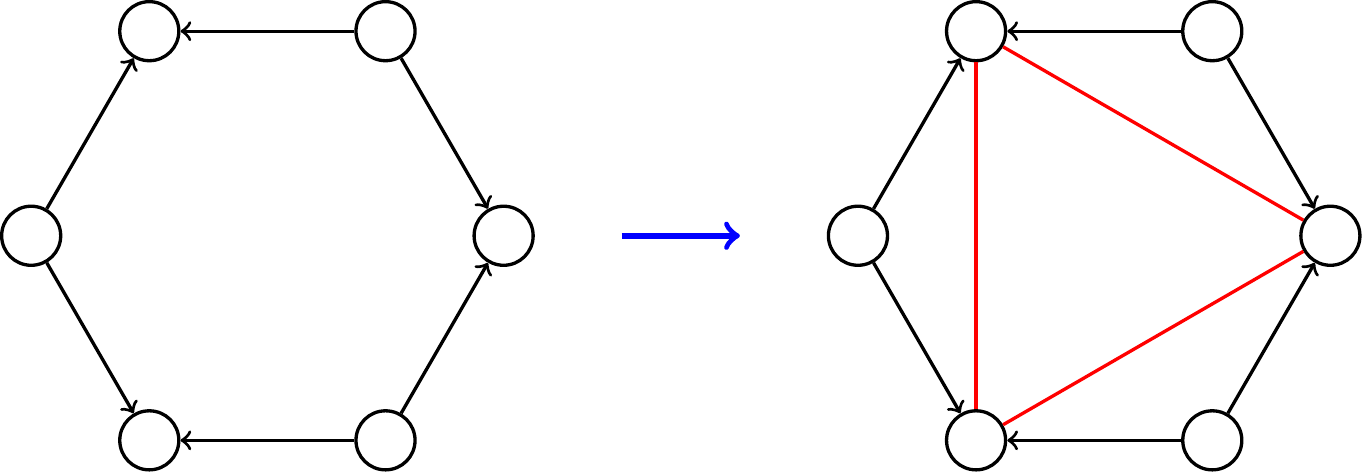}
		\end{minipage}
		\hspace{0.5cm}
		\begin{minipage}{.47\linewidth}
			\centering
			\includegraphics[width=\textwidth]{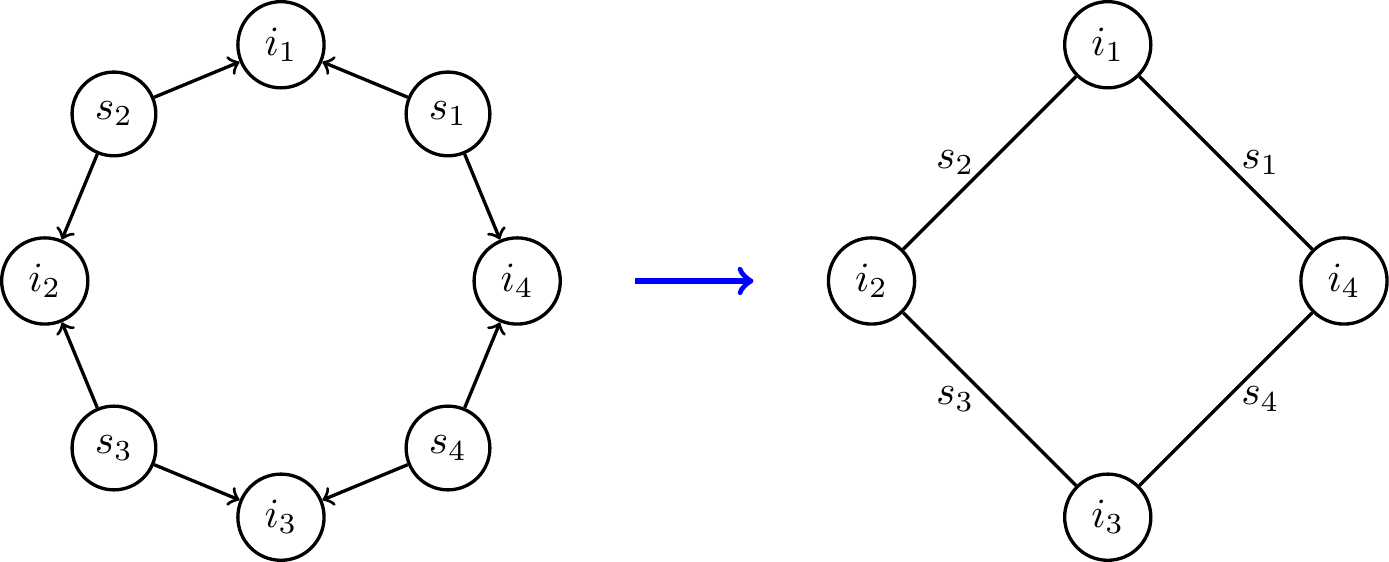}
		\end{minipage}
		\caption{(a) The $6$-cycle obstruction, this orientation has three sources intersecting with each other, this oriented pattern can not be counted in linear time in bounded degeneracy graphs. Adding an edge connecting the end-points of every out-out wedge gives a triangle. (b) An example of how the oriented $\cycle{8}$ reduces to a $\cycle{4}$, the four sources become edges connecting the intersection vertices.} 
		\label{fig:cycles}
	\end{figure}
	
	One can extend this idea more generally as follows. 
	Let $\vec{H}$ be a directed pattern, a source is any vertex with no incoming arcs, and we define an intersection vertex as any vertex that can be ``reached'' by two different sources.
	
	Suppose there are $k$ sources and $k$ intersection vertices. Suppose further that there is an ordering of the sources $\{s_0,...,s_{k-1}\}$ and the intersection vertices $\{i_0,...,i_{k-1}\}$ of $\vec{H}$ such that, for all $j$, only the sources $s_j$ and $s_{j+1}$ (taking modulo $k$ in the indexes) can both reach the vertex $i_j$.
	This is the case of the oriented $\cycle{6}$ and $\cycle{8}$ in \Fig{cycles} or of any acyclic orientation of any cycle.
	One can then construct a new graph $G'$ such that $\Sub{G'}{\cycle{k}}$ is the same as $\Hom{\vec{G}}{\vec{H}}$ (where $\Hom{\vec{G}}{\vec{H}}$ denotes the homomorphism count).

	\subsubsection{Generalizing to non-cyclic patterns}
	
	So far, the algorithmic approach only makes sense for cycle patterns. Our main contribution is a framework that generalizes this approach to count homomorphisms and subgraphs of more complex patterns. 
	
	The first part of our framework is the concept of \reducible{P} patterns. A directed pattern is \reducible{P} if we can reduce counting homomorphisms of it to counting $P$ subgraphs in a sparse hypergraph. The formal definition is technical and can be seen in \Sec{reduce}. We provide a simplified exposition in this section.
	
	Consider a directed pattern $\vec{H}$, for example the left image in \Fig{reductions}. We can replace every source with a hyperedge connecting the intersection vertices reachable by the source. The result is a graph (or hypergraph) $P$ such that $\vec{H}$ is \reducible{P}. 
	
	However, our framework allows for even more freedom: instead of looking at individual sources, we partition sources into sets of sources $S_e$. See the rightmost figure in \Fig{reductions} for an example, where the $6$ sources are divided into $4$ sets of sources. 
	The sub-patterns induced by the vertices reachable by every set of sources are ``easy'' to count using existing techniques. We can also arrange the intersection vertices into sets $I_v$ such that every set of sources will reach some sets of intersection vertices. We can obtain $P$ by replacing every set of intersection vertices with a vertex and every set of sources $S_e$ with a hyperedge $e$ that contains the vertices corresponding to the sets of intersection vertices that can be reached by $S_e$.
	
	In the example of \Fig{reductions}, the intersection of the vertices reachable from source sets forms a ``cyclic arrangement''. The intersection vertices reachable from $S_1$ are also reachable from $S_2$ and $S_4$, and similarly for the other sets of sources, giving that the pattern will be \reducible{\cycle{4}}. 

	\begin{figure}[t]
		\centering
		\begin{minipage}{.4\linewidth}
			\centering
			\includegraphics[width=\textwidth]{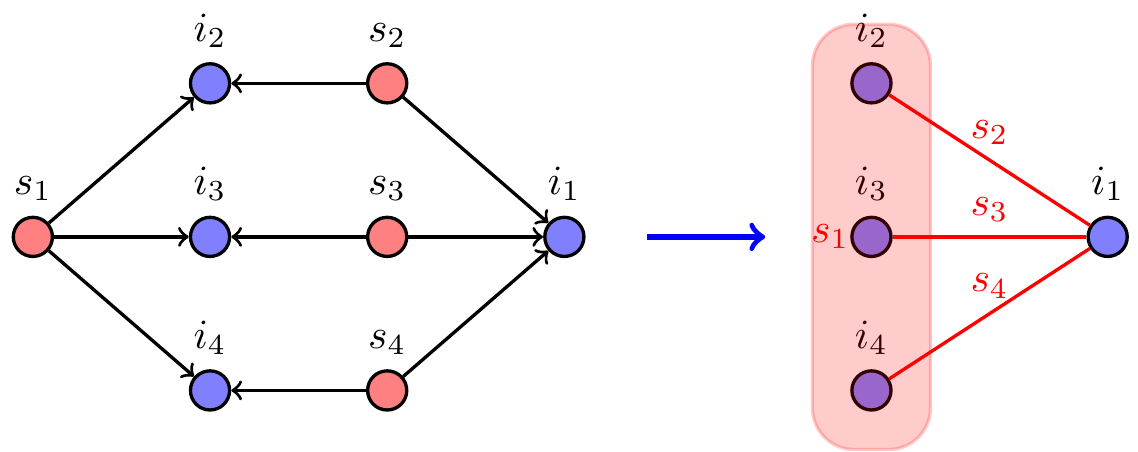}
		\end{minipage}
		\hspace{0.5cm}
		\begin{minipage}{.52\linewidth}
			\centering
			\includegraphics[width=\textwidth]{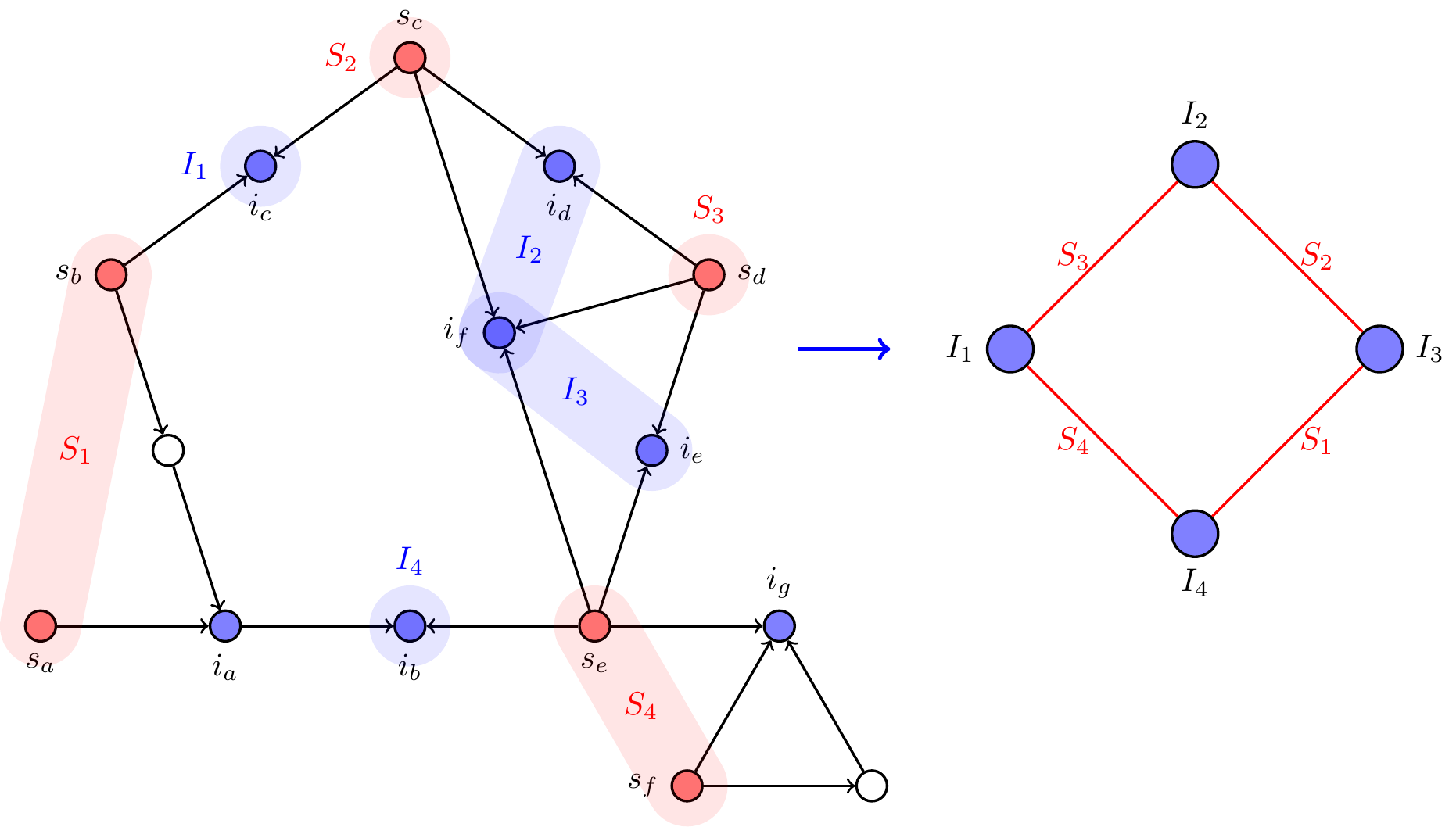}
		\end{minipage}
		\caption{Two more complex examples $P$-reducibility.} 
		\label{fig:reductions}
	\end{figure}
	
	\subsubsection{The reduced graph} \label{sec:reduced}
	
	The second part of our framework is the \emph{reduced graph}. If a pattern is \reducible{P}, then for any directed input graph $\vec{G}$,
	we can construct a colored weighted graph $\reduced{P}$ with the following property. The number of homomorphisms of the original pattern relates to the number of colorful copies of $P$ in $\reduced{P}$. 
	
	The reduced graph consists of $|V(P)|$ layers of vertices, where each layer is related to one intersection set of $\vec{H}$. Specifically, there is a vertex in the $j$-th layer for every possible image of the corresponding intersection set $I_j$ in $\vec{G}$, that is, every set of vertices in $\vec{G}$ such that $I_j$ can be mapped to it. Moreover, the vertices of every layer will have the same color. For example, for every intersection set $I_j$ in $\vec{H}$ and any map $\phi:I_j\to \vec{G}$ there is a vertex $(\phi(I_j) \mhyphen j)$ in $\reduced{P}$ with color $j$.\footnote{There could be $O(n^{|I_j|})$ such vertices, however, as we will note later, there will be at most $O(n)$ edges, so we can ignore vertices with degree $0$ and we do not need to even create them.}
	
	The edges have weights that represent the number of homomorphisms mapping the intersection sets to the corresponding images. For example, let $S$ be a source set reaching two intersection sets $I_j=\{i_j\}$ and $I_{j'}=\{i_{j'}\}$, and let $u,v \in \vec{G}$. An edge $e$ connecting $(u\mhyphen j)$ and $(v \mhyphen j')$ with weight $w=w(e)$ indicates that there are $w$ different homomorphisms mapping $\vec{H}(S)$ (the subgraph of $\vec{H}$ induced by the vertices reachable by $S$) to $\vec{G}$ that map $j$ to $u$ and $j'$ to $v$. Additionally we can show that if $\vec{G}$ has bounded outdegree, then the new reduced graph will have $O(n)$ edges.
	
	We give an example in \Fig{example_reduction}. For simplicity of exposition, the graph $\vec{G}$ is smaller than the pattern. 
	Each vertex of $\vec{G}$ has three copies (each with a different color) in the reduced graph, denoted $G_{\cycle{3}}$. Since the vertex $a$ cannot be the image
	of any intersection vertex (it has zero indegree), copies of this vertex do not appear in $G_{\cycle{3}}$. Observe that there is no mapping of an out-out wedge 
	where $b$ is one endpoint and $d$ is the other endpoint. Hence, there is no edge from a copy of $b$ to a copy of $d$. 
	
	\begin{figure}
		\centering
		\includegraphics[width=\textwidth*3/5]{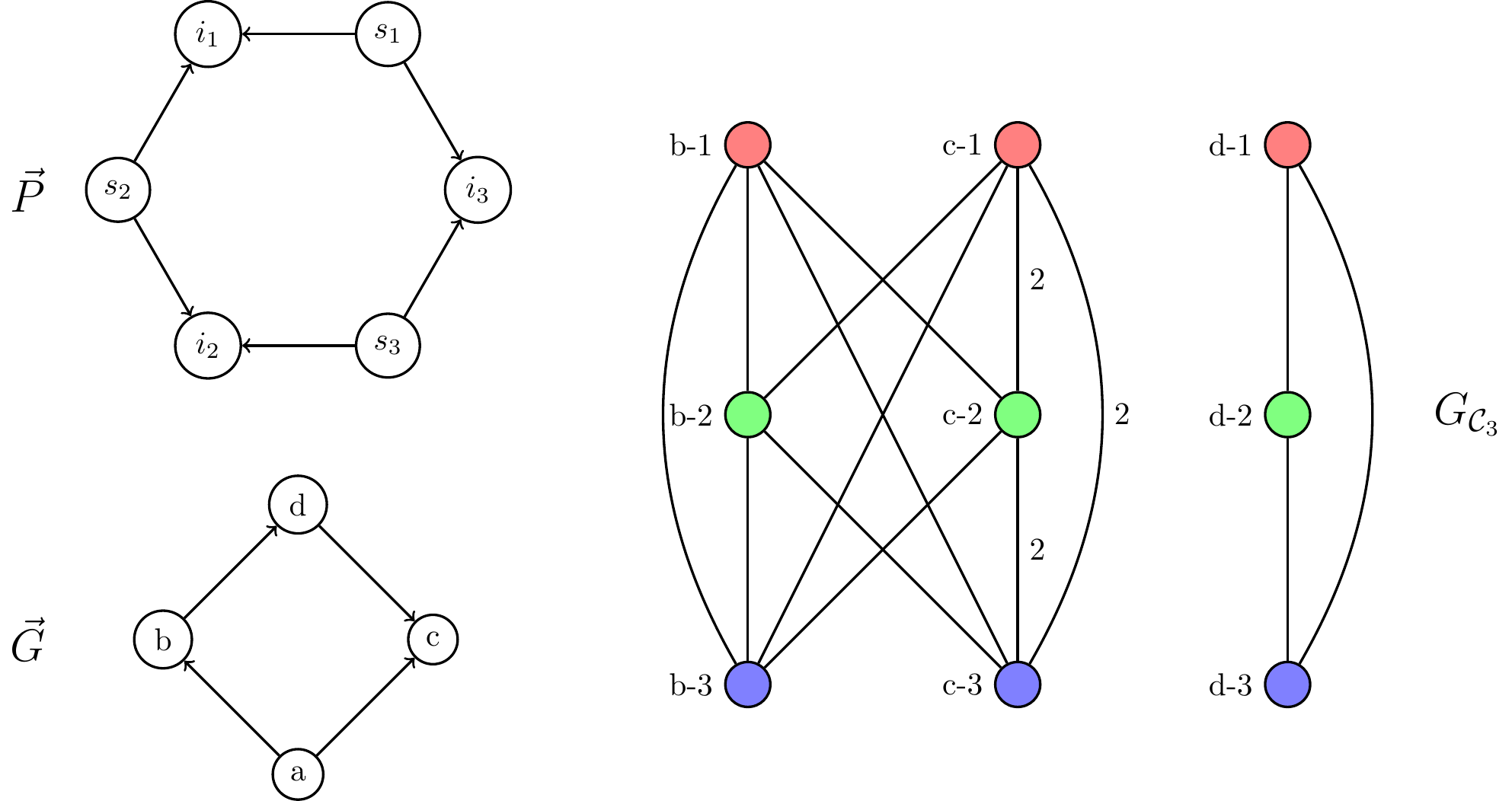}%
		\caption{An example of the construction of $\reduced{\cycle{3}}$, for pattern $\vec{H}$ and input graph $\vec{G}$. The red vertices correspond with $i_1$ in $\vec{H}$, the green ones with $i_2$ and the blue ones with $i_3$. The weight of the edges is $1$ except when indicated. For example there are two homomorphisms $\phi: \vec{H}(s_2) \to \vec{G}$ that map $i_1$ and $i_2$ to $c$, hence the edge $(c\mhyphen 1,c\mhyphen 2)$ has weight $2$. One can verify that the number of homomorphisms from $\vec{H}$ to $\vec{G}$ is equal to the sum of products of (colorful) triangles in $\reduced{\cycle{3}}$.}
		\label{fig:example_reduction}
	\end{figure}
	
	Every colorful copy of $P$ in $\reduced{P}$ correspond to fixing the positions of the intersection sets in $\vec{G}$, and the weight of each hyperedge will correspond to the number of homomorphisms mapping that portion of the graph. Hence the product of the weights of all hyperedges in each copy will give the total number of homomorphisms mapping all the intersection sets to the corresponding vertices in $\vec{G}$.
	Therefore, the total number of homomorphisms will be equal to the quantity $\WSub{\reduced{P}}{P}$, which corresponds to the sum of products of weights of colorful copies of $P$. That is:
	\begin{equation} \label{eq:wsub}
		\WSub{\reduced{P}}{P} = \sum_{P \in \colSetSub(P,\reduced{P})} \prod_{e\in E(P)} w(e)
	\end{equation}
	Here, $\colSetSub(P,\reduced{P})$ denotes the set of distinct colorful copies of $P$ in $\reduced{P}$. We are able to show that solving $\WSub{\reduced{P}}{P}$ is equivalent to counting the number of homomorphisms of $\vec{H}$ in $\vec{G}$.

	Therefore, we can count homomorphisms of \reducible{P} patterns in the same time as counting colorful copies of $P$ in the reduced graph.
	
	\begin{restatable}{lemma}{maindirected} \label{lem:main_directed}
		Let $c>1$, if there exists a $\tilde{O}(m^{c})$ algorithm that for any graph $G'$ computes $\WSub{G'}{P}$. Then, for any \reducible{P} pattern $\vec{H}$ and directed input graph $\vec{G}$ we can compute $\Hom{\vec{G}}{\vec{H}}$ in time $f(\maxoutdeg)\tilde{O}(n^{c})$, where $\maxoutdeg$ is the maximum outdegree of $\vec{G}$.
	\end{restatable}

	\subsubsection{From directed to undirected: getting homomorphisms counts}
	
	To extend our reduction framework to undirected graphs we introduce the concept of \computable{\cP} patterns. Note that different acyclic orientations of the same pattern might reduce to different graphs. We say that a pattern $H$ is \computable{\cP} if all the acyclic orientations of $H$ can be computed in linear time or are \reducible{P} for some $P\in \cP$. If for all the patterns in $\cP$ we can compute $\WSubNI{P}$ in time $O(m^c)$ then we can use \Lem{main_directed} to get a bound in the complexity of $\Hom{G}{H}$.

	For each $k\geq 6$, we find a set of hypergraphs $\cP_k$ such that all patterns with $k$ vertices are \computable{\cP_k}. These sets grow with $k$ and we have $\cP_k \subseteq \cP_{k+1}$. We give a definition of these sets in Section \Sec{others} and prove the following lemma.
	
	\begin{restatable}{lemma}{pkcomputable} \label{lem:pk_computable}
		Let $k\geq 6$, every connected pattern $H$ with $k$ vertices is \computable{\mathcal{P}_{k}}.
	\end{restatable}

	For $k=9$ we are able to show that there exists a subset $\cP^*_9$ of $\cP_9$ such that every $9$-vertex pattern is also \computable{\cP^*_9} and for every pattern $P \in \cP^*_9$ we can compute $\WSubNI{P}$ in $\tilde{O}(m^{5/3})$ time. Allowing us to show that for all patterns with $9$ vertices or less we can count the number of homomorphisms in subquadratic time.

	We can also show similar results for cycles. All orientations of cycle patterns can be reduced to cycles of half their length. This means that any $2k$-cycle and $2k+1$-cycle are \computable{\{\cycle{k}\cup...\cup \cycle{3} \}} (in these cases we will simply write \computable{\cycle{k}}).

	We can also show that for all cycles we can compute $\WSub{G}{\cycle{k}}$ in $\tilde{O}(m^{d_k})$ time, the fastest time for detecting $k$-cycles in general graphs.
	
	\begin{restatable}{lemma}{cyclecomplexity} \label{lem:cycle_complexity}
		For all $k\geq 3$, there is an algorithm that computes $\WSub{G}{\cycle{k}}$ in time $\tilde{O}(m^{d_k})$.
	\end{restatable}
	
	This means that we can compute the number of homomorphisms of $\cycle{2k}$ and $\cycle{2k+1}$ in time $\tilde{O}(m^{d_k})$, similar result to the one obtained in  GLSY~\cite{GiLeSh+23}.

	\subsubsection{Getting subgraph counts} \label{sec:subgraph}
	
	At this point, we have algorithms for computing various pattern \emph{homomorphisms}. To get subgraph counts,
	we need the inclusion-exclusion techniques of~\cite{CuDeMa17}. One can express the $H$-subgraph
	count as a linear combination of $H'$-homomorphism counts, where $H'$ is a pattern in $\Spasm(H)$.
	The $\Spasm(H)$ consists of all patterns $H'$ such that $H$ has a surjective homomorphism to $H'$.
	Thus, every pattern in the spasm has at most as many vertices as $H$. 

	Hence for any pattern $H$ with $k$ vertices, all the patterns in $\Spasm(H)$ will have at most $k$ vertices and hence they will also be \computable{\cP_k}, giving \Thm{main}.
	
	In the case of the cycles, to be able to extend the results from last section to the $Sub$ problem, we analyze the spasm of the different cycles. For $k\leq 10$ we are able to show that the patterns in the spasms of $\cycle{k}$ are also \computable{\cycle{\lfloor k/2 \rfloor}}. That combined with \Lem{cycle_complexity} implies the upper bound of \Thm{cycles}.
	
    \subsubsection{Inverting the reduction for conditional hardness} \label{sec:invert}
	
    We show that in some cases our reduction procedure is optimal. For example, counting small cycles in general graphs can be reduced to counting cycles of twice the length in graphs of degeneracy $\kappa=2$. The reduction is quite simple and just involves subdividing the edges. With a slight modification, the subdivision approach can be used to show lower bounds for odd cycles too, which gives the lower bound of \Thm{cycles}. We note that an analogous result for counting homomorphisms was shown in GLSY~\cite{GiLeSh+23}.
	
	\begin{restatable}{lemma}{lowerbound} \label{lem:lowerbound}
		Let $6\leq k \leq 10$, an $f(\degen)o(n^{d_{\lfloor k/2 \rfloor}})$ algorithm for counting $k$-cycles implies the existence of a $o(m^{d_{\lfloor k/2 \rfloor}})$ algorithm for counting $\lfloor k/2 \rfloor$-cycles. 
	\end{restatable}

	 This inversion of the reduction procedure also gives us an algorithm for counting undirected $5$-cycles in general graphs which improves on the current state of the art, and matches the complexity for $5$-cycle detection.
	
	\begin{restatable}{corollary}{fivecycle} \label{cor:5-cycle}
		There is an algorithm that, for any graph $G$, computes $\Sub{G}{\cycle{5}}$ in time $O\left(m^{d_5} \right) \approx O\left(m^{1.63}\right)$.
	\end{restatable}

	The approach of subdividing edges can be used to prove a relation between other patterns too. In the case of hypergraphs, we can replace hyperedges of arity $r$ by $r$-stars. We are able to show a strong relation between the hypergraph $\hyperthree$ depicted in \Fig{hyperthree}, and a $10$-vertex pattern. We conjecture that for this pattern we can not count the number of subgraphs in subquadratic time.
	\begin{conjecture} \label{conj:hyperthree}
		There is no $o(m^2)$ algorithm for computing $\Sub{G}{\hyperthree}$.
	\end{conjecture}
	\begin{figure}
	\centering
	\includegraphics[width=\textwidth*1/4]{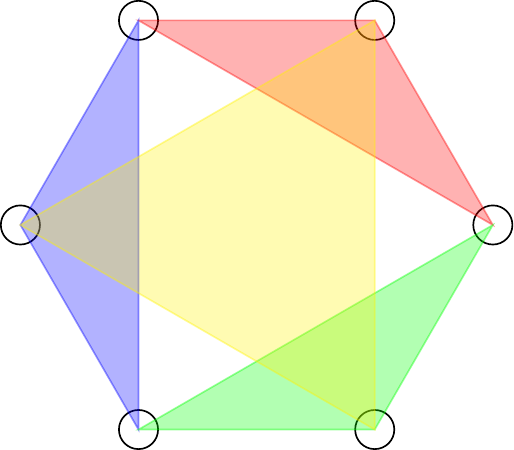}%
	\caption{The hypergraph $\hyperthree$.}
	\label{fig:hyperthree}
	\end{figure}
	This conjecture implies that there are no subquadratic algorithms for computing the number of subgraphs of all $10$-vertex patterns.
	\begin{restatable}{lemma}{hardness} \label{lem:hardness}
		If \Conj{hyperthree} holds, there is no algorithm that computes $\Sub{G}{H}$ in time $f(\degen)o(n^{2})$ for all patterns $H$ with $10$ vertices.
	\end{restatable}
	We consider it an interesting open problem to relate our conjecture with existing fine-grained complexity assumptions. Some previous works prove barriers for subquadratic \emph{listing} in general graphs~\cite{BrGo24}, but no similar results exist for counting hypergraphs.

	\subsection{Related Work}  \label{sec:related}
	
	Subgraph counting is closely tied to homomorphism counting; in some cases,
	it is more convenient to talk about the latter. Seminal work
	of Curticepean-Dell-Marx showed that the optimal algorithms for subgraph
	counting can be designed from homomorphism counting algorithms and vice versa \cite{CuDeMa17}.
	
	Much of the study of subgraph/homomorphism counting comes from
	paramterized complexity theory.
	D{\'\i}az et al ~\cite{DiSeTh02} gave a $O(2^{k}n^{tw(H)+1})$
	algorithm for determining the $H$-homomorphism count,
	where $tw(H)$ is the treewidth of $H$.
	Dalmau and Jonsson~\cite{DaJo04} proved
	that $\Hom{G}{H}$ is polynomial time solvable iff
	$H$ has bounded treewidth. Otherwise it is
	$\#W[1]$-complete. Roth and Wellnitz~\cite{RoWe20} consider 
	restrictions of both $H$ and $G$, and focus of $\#W[1]$-completeness.
	
	Tree decompositions have played an important role in subgraph counting.
	We give a brief review of the graph parameters treewidth and degeneracy.
	The notion of tree decomposition and treewidth were introduced in a seminal work by Robertson and Seymour~\cite{RoSe83,RoSe84,RoSe86}, although the concept
	was known earlier~\cite{BeBr73,Ha76}.
	
	Degeneracy is a measure of sparsity and has
	been known since the early work of Szekeres-Wilf~\cite{SzWi68}.
	We refer the reader to the short survey of Seshadhri~\cite{Se23} about degeneracy
	and algorithms. The 
	degeneracy has been exploited for subgraph counting problems in many
	algorithmic results~\cite{ChNi85,Ep94,AhNeRo+15,JhSePi15,PiSeVi17,OrBr17,JaSe17,PaSe20}.
	
	Bressan connected degeneracy to treewidth-like notation and introduced
	the concept of DAG treewidth~\cite{Br19, Br21}. 
	The main result is the following. For a pattern $H$ with $|V(H)|=k$ and 
	an input graph $G$ with $|E(G)|=m$ and degeneracy $\degen$,
	one can count $\Hom{G}{H}$ in
	$f(\degen,k)O(m^{\dtw(H)}\log m)$ time, where $\dtw(H)$ is the 
	DAG treewidth of $H$. Assuming the exponential time hypothesis~\cite{ImPaZa98}, the subgraph counting problem does not admit
	any $f(\degen,k)m^{o(\dtw(H)/\ln \dtw(H))})$ algorithm, for any positive function $f:\mathbb{N}\times \mathbb{N} \rightarrow \mathbb{N}$.
	
	Bera-Pashanasangi-Seshadhri introduced the first theory of linear time homomorphism
	counting \cite{BePaSe20}, showing that all patterns with at most $5$ vertices could be counted in linear time. It was later shown that for every pattern with no induced cycles of length $6$ or more, the number of homomorphisms could also be counted in linear time \cite{BePaSe21,BeGiLe+22}.
	
	A recent work of Komarath et al. \cite{KoKuMi+23} gave quadratic and cubic algorithms for counting cycles in sparse graphs. Gishboliner et al. gave subquadratic algorithms for homomorphism counting of cycles in bounded degeneracy graphs \cite{GiLeSh+23}.
	Bressan, Lanziger and Roth have also studied counting algorithms for directed patterns \cite{BrLaRo23}.

	\subsection{Paper Organization}
	
	In Section $3$, we define our reduction framework formally and prove the equivalence between the bounded degeneracy and the general setting. In Section $4$,
    we prove that many patterns are cycle-reducible. In Section $5$, we define the sets $\cP_k$ and complete the proof of the main theorem. Finally, in Section $6$, we show how to compute Col-WSub for cycle patterns. Due to space limitations most proofs have been deferred to the appendix, including the proofs of \Lem{lowerbound} and \Lem{hardness}.
	
	\section{Preliminaries} \label{sec:prelim}
	
	\paragraph*{Graphs, subgraphs and homomorphisms}
	
	We use $H = (V(H), E(H))$ to denote the pattern graph and $G = (V(G),E(G))$ to denote the input graph. We will use $n=|V(G)|$ and $m=|E(G)|$ for the number of vertices and edges of $G$ respectively.
	
	A homomorphism from $H$ to $G$ is a mapping $\phi: V(H) \to V(G)$ such that $\forall (u,v) \in E(H)$ we have $(\phi(u),\phi(v)) \in E(G)$. We use $\Phi(H,G)$ to denote the set of all homomorphisms from $H$ to $G$. We use $\Hom{G}{H}$ to denote the problem of counting the number of homomorphisms from $H$ to $G$, that is, computing $|\phi(H,G)|$. Similarly we use $\Sub{G}{H}$ to denote the problem of counting the number of subgraphs (not necessarily induced) of $G$ isomorphic to $H$.
	
	We say that two homomorphisms $\phi: V \to G$ and $\phi': V' \to G$ agree if for any vertex $v \in V \cap V'$ we have $\phi(v) = \phi'(v)$.
	
	The spasm of a graph is the set of all possible graphs obtained by recursively combining two vertices that are not connected by an edge, removing any duplicated edge. Using inclusion-exclusion arguments one can express the value of $\Sub{G}{H}$ as a weighted sum of homomorphism counts of the graphs in the spasm of $H$. Hence, computing $\Hom{G}{H'}$ for all $H'\in \Spasm(H)$, allows to compute $\Sub{G}{H}$, as given by the following equation: $\Sub{G}{H}  = \sum_{H' \in \Spasm(H)} f(H') \Hom{G}{H'}$.
	
	Here $f(H')$ are a series of non-zero coefficients that can be computed for each $H$. See \cite{BoChLo+06} for more details. Curticapean, Dell and Marx showed that this process is optimal \cite{CuDeMa17}, that is, the complexity of $\Sub{G}{H}$ is exactly the hardest complexity of $\Hom{G}{H'}$ for the graphs in $\Spasm(H)$.
	
	\paragraph*{Degeneracy and directed graphs}
	
	A graph $G$ is $\degen$-degenerate if every subgraph has a minimum degree of at most $\degen$. The degeneracy $\degen(G)$ of a graph is the minimum value such that $G$ is $\degen$-degenerate. There exists an acyclic orientation of a graph, called the degeneracy orientation, which has the property that the maximum outdegree of the graph is at most $\degen$ \cite{MaBe83}. Additionally, this orientation can be computed in $O(n+m)$ time. We will use $\vec{G}$ to denote the directed input graph. For a pattern $H$, we use $\Sigma(H)$ to denote the set of acyclic orientations of $H$. When orienting the input graph, every occurrence of $H$ will now appear as exactly one of its acyclic orientations, that is, $\Hom{G}{H} = \sum_{\vec{H} \in \Sigma(H)} \Hom{\vec{G}}{\vec{H}}$.
	
	We say that $v \in V(\vec{H})$ is a source if its in-degree is $0$. We use $S(\vec{H})$ to denote the set of sources of $\vec{H}$. We say a vertex $u$ is reachable from $v$ if there is a directed path connecting $v$ to $u$ and use $\Reachable_{\vec{H}}(s)$ to denote the set of vertices reachable by the source $s$. Abusing notation, for a set $S' \subseteq S(\vec{H})$ we use $\Reachable_{\vec{H}}(S')$ (or $\Reachable(S')$ if $\vec{H}$ is clear from the context) to denote the set of vertices reachable by any vertex in $S'$. We use $\vec{H}(S)$ and $\vec{H}(s)$ for the subgraphs of $\vec{H}$ induced by $\Reachable_{\vec{H}}(S)$ and $\Reachable_{\vec{H}}(s)$. 
	
	We say that a vertex $v$ of $\vec{H}$ is an intersection vertex if there are at least two distinct sources $s,s' \in S(\vec{H})$ such that $v$ is reachable by both of them, that is, $v \in \vec{H}(s) \cap \vec{H}(s')$. We use $I(\vec{H})$ to denote the set of intersection vertices. Note that $S(\vec{H}) \cap I(\vec{H}) = \emptyset$.
	
	\paragraph*{The \dagtreewidth}
	
	Bressan introduced the concept of \dagtree{} of a directed acyclic graph \cite{Br19}:
	
	\begin{definition}[\dagtree{} \cite{Br19}]
		For a given directed acyclic graph $\vec{H} = (V(\vec{H}),E(\vec{H}))$, a \dagtree{} of $\vec{H}$ is a rooted tree $T=(\cB,\cE)$ such that: 
		\begin{itemize}
			\item Each node $B \in \cB$, is a subset of the sources of $\vec{H}$, $B \subseteq S(\vec{H})$.
			\item Every source of $\vec{H}$ is in at least one node of $T$, $\bigcup_{B \in \cB} B = S(\vec{H})$.
			\item $\forall B,B_1,B_2 \in \cB$. If $B$ is in the unique path between $B_1$ and $B_2$ in $T$, then $\Reachable_{\vec{H}}(B_1) \cap \Reachable_{\vec{H}}(B_2) \subseteq \Reachable_{\vec{H}}(B)$.
		\end{itemize}
	\end{definition}
	
	The \dagtreewidth{} of a \dagtree{} $T$,  $\dtw(T)$, is the maximum size of all the bags in $T$. The \dagtreewidth{} of a directed acyclic graph $\vec{H}$ is the minimum value of $\dtw(T)$ across all valid \dagtree{} $T$ of $\vec{H}$. For an undirected graph $H$, we will have that $\dtw(H) = \max_{\vec{H} \in \Sigma(H)} \dtw(\vec{H})$.
	
	\paragraph*{Using \dagtree{} to compute homomorphisms}
	
	Bressan gave an algorithm that computes $\Hom{G}{H}$ making use of the \dagtree{} of a directed pattern \cite{Br19,Br21}. The algorithm decomposes the pattern into smaller subgraphs, computes the number of homomorphisms of every subgraph and then combines the counts using dynamic programming.
	\begin{theorem} \label{thm:bressan} \cite{Br19}
		For any pattern $H$ with $k$ vertices there is an algorithm that computes $\Hom{G}{H}$ in time $f(k, \degen) \tilde{O}(n^{\dtw(H)})$.
	\end{theorem}
	For the patterns with $\dtw(H) = 1$ we obtain an algorithm that runs in linear time when parameterized by the degeneracy of the input graph. Bera et al. showed an exact characterization of which patterns have $\dtw(H) = 1$.
	\begin{lemma} \label{lem:licl} \cite{BePaSe21}
		$LICL(H) < 6 \Leftrightarrow$ $\dtw(H) = 1$.
	\end{lemma}
	
	If we analyze the algorithm from Bressan in more detail, we can see that it can be used as a black box to obtain some more fine-grained counts. In order to understand this we need to introduce the following definition:
	
	A homomorphism $\phi'$ extends $\phi$ if for every vertex $u \in \phi$ we have $\phi(u) = \phi'(u)$. Let $\phi$ be a homomorphism from a subgraph $\vec{H}'$ of $\vec{H}$ to $\vec{G}$. We define $\extension{\vec{H}}{\vec{G}}{\phi}$ as the number of homomorphisms $\phi'$ from $\vec{H}$ to $\vec{G}$ that extend $\phi$.
	
	\begin{lemma} [\label{lem:bressan_ext} Lemma $5$ in \cite{Br21} (restated)] There exists an algorithm that, given a directed pattern $\vec{H}$ with $k$ vertices and a \dagtree{} $T$ rooted in $s$ with $\dtw(T)=1$, and a directed graph $\vec{G}$ with $n$ vertices and maximum outdegree $d$; returns, for every homomorphism $\phi:\vec{H}[s] \to \vec{G}$, the quantity $\extension{\vec{H}}{\vec{G}}{\phi}$. The algorithm runs in time $f(k,d)\tilde{O}(n)$.
	\end{lemma}
	
	\paragraph*{Hypergraphs}
	
	A hypergraph is a graph where the each edge (or hyperedge) is a subset of the vertices. The arity of a hyperedge is the number of vertices that it contains. We will only consider hypergraphs where every hyperedge has arity at least $2$. We use $E(G)$ for the set of hyperedges of $G$, where for each $e \in E(G)$ we have $e \subseteq V(G)$.We will also consider weighted hypergraphs, for a hypergraph $G$, the function $w: E(G) \to \NN$ gives the weight of each hyperedge $e$ in $G$. We use $\simplex{k}$ to denote the simplex hypergraph of arity $k$.
	
	\section{The reduction procedure} \label{sec:reduce}
	
	In this section we explain our reduction procedure. We start with the concept of $P$-reducibility. The main idea is to divide the set of sources of the directed pattern $\vec{H}$ into $|E(P)|$ different subsets of sources $S_e$, each corresponding to one hyperedge $e$ in $E(P)$, such that every source in $S(\vec{H})$ belong to exactly one subset $S_e$. We also require that the subgraphs reachable by every set of source can be counted efficiently using Bressan's algorithm.
	
	We also set $|V(P)|$ distinct subsets of intersections vertices $I_v \subseteq I(\vec{H})$, with each subset corresponding to each vertex $v$ in $V(P)$. Not all the intersection vertices of $I(\vec{H})$ must belong necessarily to one of the subsets, and we use $I^*$ to denote the set of intersection vertices that appear in at least one of the set. Moreover, different intersection sets can contain the same intersection vertex, however we require that the vertices corresponding to sets that contain the same intersection vertex in $V(P)$ induce a connected subgraph.
	
	We use $I(e)$ to denote the vertices in $I^*$ reachable by the sources in $S_e$. For every hyperedge $e$ of $P$, we require that the corresponding set of sources can reach all the vertices in the intersection sets corresponding to the vertices of $e$. Also the set of sources related to every hyperedge containing the vertex $v$ must reach all the vertices in $I_v$. These conditions ensure that we will be able to combine the homomorphism counts for each of the subgraphs $\vec{H}(S_e)$. We now present the full definition.
	
	\begin{definition} [\reducible{P}] \label{def:reducible}
	A connected DAG $\vec{H}$ is \reducible{P} if we can set:
	\begin{itemize}
		\item For every vertex $v \in V(P)$, a subset of intersection vertices $I_v \subseteq I(\vec{H})$. With $I^*=\bigcup_i I_i$. Such that for every intersection vertex $i \in I^*$, we have that the vertices $\{v : i \in I_v\} \subseteq V(P)$ induce a connected hypergraph in $P$.
		
		\item For every hyperedge $e \in E(P)$, a subset of sources $S_e \subseteq S(\vec{H})$. With $S_e \cap S_{e'} =\emptyset\ \forall e\neq e'$ and $\bigcup_{e\in E(P)} S_e = S(\vec{H})$. Such that every subset of sources $S_e$ contains a source $s_e$ with $\vec{H}(s_e) \cap I^* = \vec{H}(S_e) \cap I^* = I(e)$, and the subgraph $\vec{H}(S_e)$ admits a $\dtw=1$ \dagtree{} rooted at the source $s_e$.
		\end{itemize}
		Satisfying:
		\begin{enumerate}
			\item  For every vertex $v \in V(P)$: $I_v \subseteq I(e) \forall e \ni v$
			\item For every hyperedge $e \in E(P)$: $\bigcup_{v \in e} I_v = I(e)$
		\end{enumerate}
	
	\end{definition}

	We now define the reduced graph $\reduced{P}$.

	\begin{definition} [Reduced graph $\reduced{P}$]  \label{def:reduced_graph} 
		Given a \reducible{P} directed pattern on source sets $\{S_e : e \in E(P)\}$ and intersection sets $\{I_v : v \in V(P)\}$, we define the reduced graph $\reduced{P}$ of the directed input graph $\vec{G}$ as follows:
		\begin{itemize}
			\item For every vertex $v\in V(P)$ and every homomorphism $\phi: I_{v} \to \vec{G}$ we have the vertex $(\phi(I_{v})\mhyphen v)$ with color $v$. The vertices with the same color form the ``layers'' of $\reduced{P}$.
			\item For every hyperedge $e\in E(P)$ and for every homomorphism $\phi: I(e) \to \vec{G}$, let $\phi_v$ be the restriction of $\phi$ to $I_v$ for each vertex $v \in e$, we will have a hyperedge connecting the vertices $\{(\phi_v(I_v) \mhyphen v) : v\in e\}$ with weight $\extension{\vec{H}(S_e)}{\vec{G}}{\phi}$.
		\end{itemize}
	\end{definition}

	We use $V^{(v)}(\reduced{P})$ to refer to the vertices of $\reduced{P}$ in the $v$-th layer. 
	The number of vertices in every layer $v$ can be up to $O(n^{|I_v|})$, however we will only consider vertices that are not isolated, that is, have degree at least $1$. We can show that we can construct $\reduced{P}$ efficiently when only considering such vertices.
	
	\begin{restatable}{lemma}{construct} \label{lem:construct_reduced_graph} 
		Given a \reducible{P} pattern $\vec{H}$ and a directed graph $\vec{G}$ with maximum outdegree $\maxoutdeg$, we can construct $\reduced{P}$ in $f(\maxoutdeg)\tilde{O}(n)$ time. Additionally, the number of non-isolated vertices and the total number of hyperedges are bounded by $f(\maxoutdeg)O(n)$.
	\end{restatable}

	We can now prove the equivalence between homomorphisms of the original pattern and weighted colorful copies of the reduced hypergraph. This lemma relates the counts between the original and the reduced graph.

	\begin{restatable}{lemma}{equivalence}\label{lem:equivalence_reduction}
		$
		\Hom{\vec{G}}{\vec{H}} = \WSub{\reduced{P}}{P}
		$
	\end{restatable}
	
	Finally, we have all the tools to complete our reduction framework, giving us \Lem{main_directed}.
	
	\maindirected*
	\begin{proof}
		For any \reducible{P} pattern we can use \Lem{construct_reduced_graph} to construct the reduced $\reduced{P}$ graph in time $f(\maxoutdeg)\tilde{O}(n)$ for any input graph $\vec{G}$. This graph will have $f(\maxoutdeg)O(n)$ edges. We can then use the $\tilde{O}(m^{c})$ algorithm to compute $\WSub{\reduced{P}}{P}$ in time $f(\maxoutdeg)\tilde{O}(n^{c})$. From \Lem{equivalence_reduction} we have that $\WSub{\reduced{P}}{P}$ will be equal to $\Hom{\vec{G}}{\vec{H}}$.	
	\end{proof}

	\subsection{From directed to undirected}
	
	We introduce the concept of \computable{\cP}, which will help us give upper bounds in the complexity of undirected patterns.
	
	\begin{definition}[\computable{\mathcal{P}}] \label{def:computable}
		Let $\mathcal{P}$ be a set of hypergraphs. We say that a pattern $H$ is \computable{\mathcal{P}} if every acyclic orientation $\vec{H} \in \Sigma(H)$ has either \dagtreewidth{} of $1$ or there exists a hypergraph $P \in \mathcal{P}$ such that $\vec{H}$ is $P$-reducible.
	\end{definition}
	
		If a pattern $H$ is \computable{\mathcal{P}} and $\mathcal{P}$ is only formed by cyclic patterns we will instead write \computable{\cycle{l}}, where $l$ is the length of the largest cycle in $\mathcal{P}$. The complexity of computing homomorphisms of \computable{\mathcal{P}} patterns will be dominated by the hardest complexity for computing $\WSubNI{P}$.
	
	\begin{restatable}{lemma}{computablelemma} \label{lem:computable}
				Let $\mathcal{P}$ be a set of hypergraphs, if for every hypergraph $P \in \mathcal{P}$ there is an algorithm that computes $\WSubNI{P}$ in time $\tilde{O}(m^c)$, then for any input graph $G$ with degeneracy $\degen$ and any \computable{\mathcal{P}} pattern $H$, there is an algorithm that computes $\Hom{G}{H}$ in time $f(\degen)\tilde{O}(n^c)$.
		\end{restatable}

	\section{Reducing to cycles} \label{sec:reductions}
	
	We first focus on patterns that can be reduced to counting cycles. We start by introducing the following two lemmas, showing that directed patterns with either few sources or few intersection vertices are \reducible{\cycle{3}}.
	
	\begin{restatable}{lemma}{threesources} \label{lem:3sources}
		Every directed acyclic pattern $\vec{H}$ with at most $3$ sources is either \reducible{\cycle{3}} or $\dtw(\vec{H})=1$.
	\end{restatable}

	\begin{restatable}{lemma}{threeintersections} \label{lem:3intersections}
		Every directed acyclic pattern $\vec{H}$ with at most $3$ intersection vertices is either \reducible{\cycle{3}} or $\dtw(\vec{H})=1$. 
	\end{restatable}
	
	Using this two lemmas we can prove that all $6$ and $7$-vertex patterns are \computable{\cycle{3}}.
	
	\begin{restatable}{lemma}{sixseven} \label{lem:sixseven}
		Every $6$ and $7$-vertex undirected pattern $H$ is \computable{\cycle{3}}.
	\end{restatable}

	Additionally, the patterns in the spasm will always have less vertices, and hence all patterns in the spasms of all $6$ and $7$-vertex patterns will also be \computable{\cycle{3}}. This fact, together with \Lem{cycle_complexity} gives the following.

	\begin{corollary} \label{cor:sixseven}
		Let $H$ be a pattern with $6$ or $7$ vertices. For any input graph $G$, we can compute $\Hom{G}{H}$ and $\Sub{G}{H}$ in time $f(\degen)\tilde{O}(n^{d_3}) \approx f(\degen)O(n^{1.41})$.
	\end{corollary}
	
	We can also show that the acyclic orientations of cycle patterns are always \reducible{\cycle{k}} for some $k$ at most half of the length of the cycle. This is equivalent to the result in \cite{GiLeSh+23}, but expressed using our reducibility framework.
	
	\begin{restatable}{lemma}{cyclehom}  \label{lem:cycle_hom}
		For all $k\geq3$, $\cycle{2k}$ and $\cycle{2k+1}$ are \computable{\cycle{k}}.
	\end{restatable}

	Moreover, we can also show that all patterns in the spasms of cycles up to length $10$ are also cycle-computable.
	
	\begin{lemma} \label{lem:spasms}
		\begin{itemize}
			\item All the patterns in $\Spasm(\cycle{6})$ and $\Spasm(\cycle{7})$ are \computable{\cycle{3}}.
			\item All the patterns in $\Spasm(\cycle{8})$ and $\Spasm(\cycle{9})$ are \computable{\cycle{4}}.
			\item All the patterns in $\Spasm(\cycle{10})$ are \computable{\cycle{5}}.
		\end{itemize}
	\end{lemma}

	This lemma allows us to prove the upper bound of \Thm{cycles}.
	\begin{restatable}{lemma}{uppercycle} \label{lem:uppercycle}
		For all $6\leq k\leq 10$, there is an algorithm that computes $\Sub{G}{\cycle{k}}$ in time $f(\degen)O(n^{d_{\lfloor k/2\rfloor}})$.
	\end{restatable}
	
	\section{Reducing to other patterns} \label{sec:others}
	
	Consider the set of directed patterns with $8$ vertices. Using the results of the previous section we can show that most of the orientations will either admit a \dagtree{} with $\tau=1$ or will be \reducible{\cycle{3}}. However, if a pattern $\vec{H}$ has $4$ sources and $4$ intersection vertices then it might not be cycle-reducible. Instead we might need to reduce to some hypergraphs, like in \Fig{reductions}. The following definitions will help us determining which patterns we will need to reduce to for patterns with at least $8$ vertices.
	
	\begin{restatable}{definition}{calp} [$\mathcal{P}_{i,s}$, $\mathcal{P}_k$] \label{def:group}
		We define $\mathcal{P}_{i,s}$ as the set of hypergraphs $P$ with $i$ vertices and $s$ hyperedges such that:
		\begin{enumerate}
			\item Every vertex has degree at least $2$.
			\item Every hyperedge contains at least $2$ vertices.
			\item No hyperedge is a subset of any other hyperedge.
			\item For every pair of distinct vertices $u,v \in V(P)$ the set of hyperedges containing $u$ can not be equal or a subset of the set of hyperedges containing $v$.
		\end{enumerate}
		For any $k\geq 7$, we define $\mathcal{P}_{k}$ recursively as the union of $\cP_{k-1}$ and all sets $\mathcal{P}_{i,s}$, with $i+s = k$ and $i,s\geq 4$, with $\cP_6=\{\cycle{3}\}$.
	\end{restatable}

	We can prove that for any $k$, patterns with $k$ vertices will reduce to some pattern in $\cP_k$. This was stated earlier as \Lem{pk_computable}.
	\pkcomputable*

	In order to prove \Thm{main}, we show exactly which patterns form $\cP_9$.
	
	\begin{lemma} \label{lem:nine_content}
		$
			\cP_9 = \{\cycle{3},\cycle{4},\diamondgraph,\simplex{3},\hyperone, \hypertwo\}
		$
	\end{lemma}
	\begin{figure}[t]
	\centering
	\begin{minipage}{.15\linewidth}
		\centering
		\includegraphics[width=\textwidth]{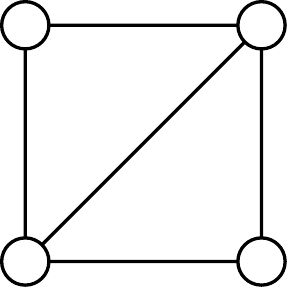}
	\end{minipage}
	\hspace{1cm}
	\begin{minipage}{.15\linewidth}
		\centering
		\includegraphics[width=\textwidth]{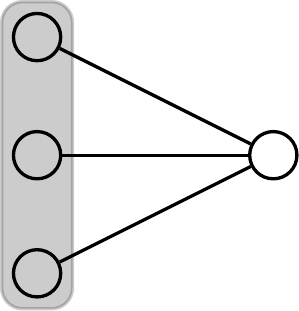}
	\end{minipage}
	\hspace{1cm}
	\begin{minipage}{.2\linewidth}
		\centering
		\includegraphics[width=\textwidth]{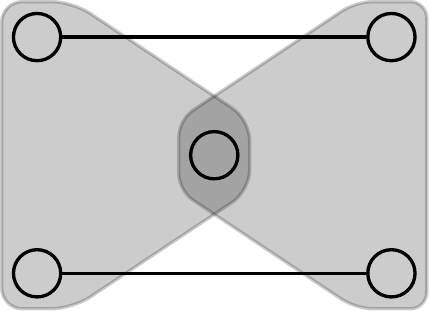}
	\end{minipage}
	\caption{The diamond graph $\diamondgraph$, the hypergraph $\hyperone$ and the hypergraph $\hypertwo$.} 
	\label{fig:hypers}
	\end{figure}
	Where, $\diamondgraph$ is the diamond pattern, $\simplex{3}$ the $3$-simplex, $\hyperone$ and $\hypertwo$ are the two hypergraphs shown in \Fig{hypers}. It turns out that simplex-reducible patterns are also cycle-reducible. Hence we can set $\cP^*_9 = \cP_9 \setminus \simplex{3}$ and show that every \computable{\cP_9} pattern is also \computable{\cP^*_9}.
	
	\begin{restatable}{lemma}{ninestar} \label{lem:nine_star}
		If a pattern $H$ is \computable{\cP_9}, then it is also \computable{\cP^*_9}
	\end{restatable}

	We can show that all the hypergraphs in $\cP^*_9$ can be counted in subquadratic time.
	
	\begin{restatable}{lemma}{allnine} \label{lem:allnine}
			For any weighted colored hypergraph $G$ with $m$ edges, there is an algorithm that computes $\WSub{G}{P}$ for all patterns $P \in \cP^*_9$ in time $\tilde{O}(m^{5/3})$.
	\end{restatable}

	Finally we can prove the main theorem.
	
	\main*
	\begin{proof}
		Let $H$ be a pattern with $9$ or less vertices. From \Lem{pk_computable} we have that $H$ is \computable{\mathcal{P}_9}, using \Lem{nine_star} we will have that it is also \computable{\cP^*_9}. Additionally, from \Lem{allnine} we have that for all hypergraphs $P\in \cP^*_9$ we can compute $\WSub{G}{P}$ in $\tilde{O}(m^{5/3})$ time. This together with \Lem{computable} gives that we compute $\Hom{G}{H}$ in $f(\degen)\tilde{O}(n^{5/3})$ time.
		
		All the graphs $H'$ in the Spasm of $H$ have also at most $9$ vertices, hence we can compute $\Hom{G}{H'}$ for them and use inclusion-exclusion to obtain the value of $\Sub{G}{H}$ in total time $f(\degen)\tilde{O}(n^{5/3})$.
	\end{proof}

	\section{Counting cycles} \label{sec:wsub}
	
	We adapt the two algorithms for counting weighted homomorphisms of cycles shown in \cite{GiLeSh+23} for computing $\WSubNI{\cycle{k}}$. The first is a combinatorial algorithm that matches the complexity of detecting directed cycles combinatorially~\cite{AlYuZw97}. The second is a matrix multiplication based algorithm which adapts the algorithm from \cite{YuZw04}. The complexity of this algorithm for counting $\cycle{k}$ is given by the value $c_k$, the exact values of $c_k$ for $k\geq 6$ are not known, but the following upper bound holds~\cite{DaVuWi19}:
	\begin{equation}
	\begin{split} 
		&c_k \leq \frac{\omega(k+1)}{2\omega+k-1}\ &\text{if $k$ is odd} \qquad\qquad
		&c_k \leq \frac{\omega k - 4/k}{2 \omega + k - 2 - 4/k}\ &\text{if $k$ is even}
	\end{split}
	\end{equation}
	Where $\omega$ is the matrix multiplication exponent.
	\begin{lemma} \label{lem:wsub_cycles}
		Let $G$ be a colored weighted graph with $m$ edges. For all $k>3$:
		\begin{itemize}
			\item There is a combinatorial algorithm that computes $\WSub{G}{\cycle{k}}$ in time $\tilde{O}(m^{2-1/\lceil k/2 \rceil })$.
			\item There is an algorithm that computes $\WSub{G}{\cycle{k}}$ in time $\tilde{O}(m^{c_k})$.
		\end{itemize}
	\end{lemma}

	For any $k$ we use $d_k$ for the fastest of the two algorithms, similar to \cite{GiLeSh+23}. \Lem{cycle_complexity} follows directly from \Lem{wsub_cycles} and the following equation:
	\begin{equation} \label{eq:dk}
		d_k = min(2-1/\lceil k/2 \rceil,c_k)
	\end{equation}

	Note that $d_k < 2$ for all $k$, hence we have subquadratic algorithms for all cycles. For $k<6$ and using the fastest matrix multiplication exponent known to date $\omega = 2.371339$ \cite{AlDuWi+25}, we have that the matrix multiplication algorithm is faster and we get the following approximate values of $d_k$: $d_3 \approx 1.41$, $d_4 \approx 1.48$ and $d_5 \approx 1.63$.
	
	\bibliography{subgraph_counting_doi}
	
	\newpage
	\appendix
	\section{Proofs of Section \ref{sec:reduce}}
	
	\subsection{Proof of \Lem{construct_reduced_graph}}
	
	First, we show how to compute homomorphism extensions of the different subgraphs $\vec{H}(S_e)$ of \reducible{P} patterns, this is a necessary step in order to construct the reduced graph.
	
	\begin{lemma} \label{lem:extension_subgraphs}
		Let $\vec{H}$ be a \reducible{P} directed pattern and $\vec{G}$ a directed input graph with $n$ vertices and maximum outdegree $\maxoutdeg$. Let $\{S_e : e\in E(P)\}$ and $\{I_v : v\in V(P)\}$ be the sets of sources and intersections that achieve the $P$-reducibility. For every set of sources $S_e$ we can compute $\extension{\vec{H}(S_e)}{\vec{G}}{\phi}$ for every $\phi: I(e) \to \vec{G}$ in $f(\maxoutdeg)\tilde{O}(n)$ time. Additionally, there are at most $f(\maxoutdeg)O(n)$ homomorphisms $\phi$ with non-zero extension.
	\end{lemma}
	\begin{proof}
		Fix a set of sources $S_e$. From the definition of \reducible{P} we have that $\vec{H}(S_e)$ admits a \dagtree{} with $\dtw=1$ rooted at $s_e$. Therefore we can use Bressan's results in \Lem{bressan_ext} to compute the quantity $\extension{\vec{H}(S_e)}{\vec{G}}{\phi}$ for each homomorphism $\phi: \vec{H}(s_e) \to \vec{G}$. Note that from \Def{reducible}, we have that $I(e) \subseteq \vec{H}(s_e)$, hence it is possible to iterate over the homomorphisms from $\vec{H}(s_e)$ and aggregate the counts in order to obtain, for each homomorphism $\phi': I(e)\to \vec{G}$ the quantity $\extension{\vec{H}(S_e)}{\vec{G}}{\phi'}$, that is, the number of homomorphisms of $\vec{H}(S_e)$ in $\vec{G}$ that map the vertices in $I(e)$ to some specific set of vertices.
		
		The number of homomorphisms from $\vec{H}(s_e)$ to $\vec{G}$ is bounded by $O(n \maxoutdeg^k)$. Hence the call to Bressan's and the aggregation will take a total of $f(\maxoutdeg)\tilde{O}(n)$ time.
		
		Additionally, each homomorphism from $\vec{H}(s_e)$ to $\vec{G}$ can contribute to exactly one homomorphism from $I(e)$ to $\vec{G}$, therefore the total number of such homomorphisms with non-zero extension will be at most $f(\maxoutdeg)O(n)$.
	\end{proof}
	
	We can now prove \Lem{construct_reduced_graph}.
	
	\construct*
	\begin{proof}[Proof of \Lem{construct_reduced_graph}]
		Let $\{S_e : e\in E(P)\}$ and $\{I_v : v\in V(P)\}$ be the sets of sources and intersection vertices that achieve the $P$-reducibility. Using \Lem{extension_subgraphs} we can compute, for every homomorphism $\phi: I(e) \to \vec{G}$ the quantity $\extension{\vec{H}(S_e)}{\vec{G}}{\phi}$, in $f(\maxoutdeg)\tilde{O}(n)$ time. Additionally we know that for each $S_e$ there are at most $f(\maxoutdeg)O(n)$ homomorphisms $\phi$ with non-zero extension.
		
		Initialize $\reduced{P}$ as the empty graph. Now, for every $e\in E(P)$, and for every $\phi: I(e) \to \vec{G}$ with non-zero extension, let $\phi_v$ be the restriction of $\phi$ to $I_v$ for each $v \in e$. If not present in the graph, create the vertex $u_v = (\phi_v(I_v)\mhyphen(v))$ for each $v \in e$. Then add a hyperedge connecting the vertices $\{ u_v : v \in e\}$ with weight $\extension{\vec{H}(S_e)}{\vec{G}}{\phi}$. The resulting graph is exactly $\reduced{P}$.
		
		The complexity of this process is dominated by computing the extensions as in \Lem{extension_subgraphs}, hence the graph can be constructed in $f(\maxoutdeg)\tilde{O}(n)$ time. Also, the number of hyperedges is equal to the number of homomorphisms with non-zero extension, which again using \Lem{extension_subgraphs} will be bounded by $f(\maxoutdeg)O(n)$. Finally, the number of vertices added to the graph is at most the number hyperedges times the maximum arity of $P$, which is bounded by $k$, and hence it will be bounded also by $f(\maxoutdeg)O(n)$.
	\end{proof}

	\subsection{Proof of \Lem{equivalence_reduction}}
	First we show that each colorful copy of $P$ in $\reduced{P}$ can be related to $|V(P)|$ homomorphisms that do not disagree in any vertex.
	
	\begin{claim} \label{clm:connected}
		Let $P' \in \colSetSub(P,\reduced{P})$ with vertex set $\{u_v : v \in V(P)\} = V(P')$. Consider the homomorphisms $\{\phi_v\}$ corresponding to each vertex $u_v \in V(P')$ such that $u_v = (\phi_v(I_v)\mhyphen v)$. Every pair of homomorphisms will agree with each other.
	\end{claim}
	\begin{proof}
		We prove by contradiction. Consider the case that for some $P' \in \colSetSub(P,\reduced{P})$ we have two vertices $u_v$ and $u_{v'}$ for which the corresponding homomorphisms $\phi_v$ and $\phi_{v'}$ disagree in the value of $i$, that is, $\exists i \in I_v \cap I_{v'}$ such that $\phi_{v}(i) \neq \phi_{v'}(i)$.
		
		From \Def{reducible} we have that the vertices $v \in V(P)$ for which $I_v \ni i$ induce a connected hypergraph in $P$. Let $P_i$ be that hypergraph. Note that $v,v' \in P_i$, hence there must exist a sequence of vertices $\{v, v_1, v_2 ... v'\}$  from $v$ to $v'$ in $P_i$ such that there is a hyperedge in $P$ containing every consecutive pair of the sequence, otherwise, $P_i$ is not connected and we reach a contradiction.
		
		Now, consider the vertices $u_v, u_{v_1},..., u_{v'} \in V(P')$ corresponding to each vertex of the sequence. For each consecutive pair $u_{v_l},u_{v_{l+1}}$ in the sequence, we have that there is a hyperedge $e$ in $\reduced{P}$ that contains both $u_{v_l}$ and $u_{v_{l+1}}$. But this implies that the homomorphisms $\phi_{v_l}$ and $\phi_{v_{l+1}}$ corresponding to those vertices can not disagree in the value of $i$, as they are both obtained by restricting the same homomorphism $\phi: I(e) \to \vec{G}$ to $I_{v_l}$ and $I_{v_{l+1}}$ respectively.
		
		This implies that for every pair of vertices in the sequence the corresponding homomorphisms must agree. However, this also implies that all the homomorphisms in the sequence must map $i$ to the same vertex, reaching a contradiction.
	\end{proof}
	
	We can now prove \Lem{equivalence_reduction}. This lemma relates the counts between the original and the reduced graph.
	
	\equivalence*
	\begin{proof}
		
		Let $\colSetSub(P,\reduced{P})$ be the collection of appearances of colorful $P$ in $\reduced{P}$. Because of how we set the colors in the reduced graph, every such appearance will have one vertex in each of the $|V(P)|$ layers of $\reduced{P}$.
		
		Let $\Phi(I^*,\vec{G}) = \{\phi: I^* \to \vec{G}\}$. Every homomorphism from $\vec{H}$ to $\vec{G}$ will extend exactly one of the homomorphisms $\phi$ in $\Phi(I^*,\vec{G})$. Hence we can write:
		
		\begin{equation} \label{eq:step0}
			\Hom{\vec{G}}{\vec{H}} = \sum_{\phi \in \Phi(I^*,\vec{G})} \extension{\vec{H}}{\vec{G}}{\phi}
		\end{equation}
		
		 Every such homomorphism $\phi$ in $\Phi(I^*,\vec{G})$ corresponds exactly with a set of homomorphisms $\phi_v \in \phi(I_v, \vec{G})$ for each $v \in V(P)$, such that $\phi = \bigcup_{v \in V(P)} \phi_v$. Note that all such homomorphisms must agree which each other. Moreover, let $\Phi_\times = \bigtimes_{v \in V(P)}{\Phi(I_v,\vec{G})}$ be the collection of all such possible sets and $\Phi'_\times$ the restriction of $\Phi_\times$ to sets where all the homomorphisms agree with each other. Every set $\{\phi_v\} \in \Phi'_\times$ will correspond with exactly one homomorphism $\phi$ in $\Phi(I^*,\vec{G})$. Hence we can express the previous expression as follows:
		
		\begin{equation} \label{eq:step1}\sum_{\phi \in \Phi(I^*,\vec{G})} \extension{\vec{H}}{\vec{G}}{\phi} = \sum_{\{\phi_v \} \in \Phi'_\times }
			\extension{\vec{H}}{\vec{G}}{\bigcup_{v \in V(P)} \phi_v}
		\end{equation}
		
		Now, consider the subgraphs $\vec{H}(S_e)$ for every $e \in E(P)$. Note that these subgraphs only intersect with each other on the vertices of $I^*$, and $I(e) = I^* \cap \vec{H}(S_e)$. For a hyperedge $e \in E(P)$ and a homomorphism $\phi \in \Phi(I^*,\vec{G})$, let $\phi_e$ be the restriction of $\phi$ to the vertices of $I(e)$. We can then write $\extension{\vec{H}}{\vec{G}}{\phi}$ as the product of extensions $\extension{\vec{H}(S_e)}{\vec{G}}{\phi_e}$. Moreover, when expressing $\phi$ as a set $\{\phi_v \} \in \Phi'_\times$, every $\phi_e$ will be equal to the union of $\phi_v$ for all $v \in e$. Hence for any such set we will have:
		
		\begin{equation}\label{eq:step2}
			\extension{\vec{H}}{\vec{G}}{\bigcup_{v \in V(P)} \phi_v} = 
			\prod_{e\in E(P)} \extension{\vec{H}(S_e)}{\vec{G}}{\bigcup_{v \in e} \phi_v}
		\end{equation}
		
		Combining \Eqn{step0}, \Eqn{step1} and \Eqn{step2} we get:
		
		\begin{equation} \label{eq:step3}
			\Hom{\vec{G}}{\vec{H}} = 
			\sum_{\{\phi_v \} \in \Phi'_\times}
			\prod_{e\in E(P)} \extension{\vec{H}(S_e)}{\vec{G}}{\bigcup_{v \in e} \phi_v}
		\end{equation}
		
		Note that every homomorphisms $\phi_v$ in the sets of $\Phi'_\times$ corresponds with the vertex $(\phi_v(I_v)\mhyphen v)$ of $\reduced{P}$. Additionally, by \Def{reduced_graph} the value of $\extension{\vec{H}(S_e)}{\vec{G}}{\bigcup_{v \in e} \phi_v}$ corresponds with the weight of the hyperedge connecting the vertices $\{(\phi_v(I_v)\mhyphen v) : v\in e\}$ of $\reduced{P}$. Thus, we can rewrite the previous expression, by iterating over vertices of the different layers, as long as their corresponding homomorphisms agree with each other. Let $V(\Phi'_\times)$ be the collection of sets of vertices $\{u_v\}$ corresponding to each $\{\phi_v\} \in \Phi'_\times$. We have:
		
		\begin{equation}\label{eq:step4}
				\sum_{\{\phi_v \} \in \Phi'_\times}
				\prod_{e\in E(P)} \extension{\vec{H}(S_e)}{\vec{G}}{\bigcup_{v \in e} \phi_v} \\
				= \sum_{\{u_v \} \in V(\Phi'_\times)} 
				\prod_{e\in E(P)} w(\{u_v: v\in e\})
		\end{equation}
		
		For a set of vertices $\{u_v\}$ to have non-zero total product, we will require that all the hyperedges $\{u_v: v\in e\}$ for $e\in E(P)$ have non-zero weight. Hence, the subgraph of $\reduced{P}$ induced by $\{u_v\}$ must contain a copy of $P$. Additionally, every copy of $P$ in $\reduced{P}$ will correspond to some set of vertices $\{u_v \} \in V(\Phi'_\times)$, as by \Clm{connected} the vertices of every copy must agree in their corresponding homomorphisms. Given this, together with \Eqn{step3} and \Eqn{step4} we get:
		
		\begin{equation}
			\Hom{\vec{G}}{\vec{H}}  
			= \sum_{\{u_v \} \in V(\Phi'_\times)} 
			\prod_{e\in E(P)} w(\{u_v: v\in e\})
			=\sum_{P' \in \colSetSub(P,\reduced{P})} \prod_{e'\in E(P')} w(e')
		\end{equation}
		
		Which by the definition of $\WSubNI{P}$ (Equation \Eqn{wsub}) is precisely $\WSub{\reduced{P}}{P}$.
		
	\end{proof}

	\subsection{Proof of \Lem{computable}}
	
	\computablelemma*
	\begin{proof}
		First we compute the degeneracy orientation $\vec{G}$ of $G$, this will take $O(n+m)$ time. Note that the number of edges is bounded by $f(\degen)O(n)$, the maximum outdegree of $\vec{G}$ will be $d = O(\degen)$.
		Let $H$ be a \computable{\cP} pattern. Every acyclic orientation $\vec{H} \in \Sigma(H)$ either has a \dagtreewidth{} of $1$ ot it is \reducible{P} for some hypergraph $P\in \cP$. In the first case we can use \Thm{bressan} to compute $\Hom{\vec{G}}{\vec{H}}$ in $f(\degen)\tilde{O}(n)$ time. In the second case we can use \Lem{main_directed} to compute $\Hom{\vec{G}}{\vec{H}}$ in time $f(d)\tilde{O}(n^c)$. We can then aggregate the counts for all the orientations, obtaining $\Hom{G}{H}$ in time $f(\degen)\tilde{O}(n^c)$.
	\end{proof}

	\section{Proofs of Section \ref{sec:reductions}}

	\subsection{Proof of \Lem{3sources}}
		\threesources*
		\begin{proof}
		First, if the pattern has only $1$ or $2$ sources, then one can construct a \dagtree{} $T$ of $\vec{H}$ with $\dtw=1$ by putting every source in its own node of $T$.
		
		Now we consider the case where $\vec{H}$ has $3$ sources. Let $s_1$, $s_2$ and $s_3$ be the three sources. Consider the case where there is a pair of sources that do not intersect with each other, let wlog $s_1,s_3$ be those sources, that is, $\Reachable(s_1) \cap \Reachable(s_3) = \emptyset$. We can construct a \dagtree{} of $\vec{H}$ by putting every source in its own bag an putting $s_2$ in between $s_1$ and $s_3$. This will be a valid \dagtree{} with $\dtw=1$.
		
		Otherwise, every pair of sources intersect each other. Let $S_i = \{s_i\}$ and set $I_i = \vec{H}(S_i) \cap \vec{H}(S_{i+1})$. We show that these sets give a valid $\cycle{3}$-reduction. First, the two sets are valid sets, the source sets are non-overlapping and have $\dtw=1$ (as there is only one source in each set). While the intersections are all adjacent to each other, so intersection sets containing the same vertex will form a continuous sequence. Also, we can check that the two conditions from \Def{reducible} are satisfied directly by how we constructed the sets. Hence the pattern will be \reducible{3}\footnote{Note that it is possible that some of these patterns also have $\dtw=1$.}. 
	\end{proof}

	\subsection{Proof of \Lem{3intersections}}
	\threeintersections*
	\begin{proof}
		If there are $1$ or $2$ intersection vertices then one can construct a valid \dagtree{} $T$ with $\dtw(T)=1$ by finding a source that reaches all the intersection vertices (one must exist or the pattern is not connected) and setting it as the root of $T$ and then connecting all the other sources directly to the root as singleton bags. For every pair of sources, their intersection will be a subset of the vertices reachable by the root and hence we will have a valid \dagtree{}.
		
		If we have $3$ intersection vertices we need to consider the following cases. If there is a source that can reach the $3$ intersection vertices then the previous construction will give a valid \dagtree{} with $\dtw=1$. Hence, we consider the case where every source at most reach $2$ of the intersection vertices. Let $\{u,v,w\} = I(\vec{H})$ be the three intersection vertices, and consider the three pairs $u\mhyphen v,v\mhyphen w,u\mhyphen w$. For every source $s \in S(\vec{H})$ we will have that $s$ either reaches one of the intersection vertices or one of the three pairs. If only two of the pairs are reachable by the sources then we can construct a \dagtree{} with $\dtw=1$: Let $u\mhyphen v,v\mhyphen w$ be the two pairs that appear, and let $s_1$ and $s_2$ be two sources with $u,v \in \Reachable(s_1)$ and $v,w \in \Reachable(s_2)$, set $s_1$ as the root of the tree and $s_2$ as one of its children, connect every other source to either $s_2$ if they can reach $w$ and to $s_1$ otherwise.
		
		If the three pairs are reachable, then we can show a valid $\cycle{3}$-reduction. Let $s_1,s_2,s_3$ be three sources reaching a different pair of intersection vertices, specifically, let $u,v \in \Reachable(s_1)$, $v,w \in \Reachable(s_2)$ and $u,w \in \Reachable(s_3)$. We put $s_1$ in $S_1$, $s_2$ in $S_2$ and $s_3 \in S_3$, for every other source, if they reach two intersection vertices then we put it in the set which contains a source reaching the same two vertices, if they reach only one intersection vertex then we arbitrarily add it to one of the two sets of sources that reach that vertex. We set $I_1=\{v\}$, $I_2=\{w\}$ and $I_3=\{u\}$. We can verify that this is a valid $\cycle{3}$-reduction, as it will satisfy all the properties from \Def{reducible}.
	\end{proof}

	\subsection{Proof of \Lem{sixseven}}
	\sixseven*	
	\begin{proof}
		Fix a pattern $H$, any acyclic orientation $\vec{H} \in H$ will either have at most $3$ sources or $3$ intersection vertices. In the first case we can use \Lem{3sources} to show that $\vec{H}$ is \reducible{\cycle{3}} or has $\dtw=1$, similarly in the second case we can use \Lem{3intersections}. Hence, all acyclic orientations of $H$ are either \reducible{\cycle{3}} or have $\dtw=1$, which implies that $H$ is \computable{\cycle{3}}.
	\end{proof}

	\subsection{Proof of \Lem{cycle_hom}}
	\cyclehom*
	\begin{proof}
		Fix any $k$, we want to show that all acyclic orientations of $\cycle{2k}$ and $\cycle{2k+1}$ have $\dtw=1$ or are \reducible{\cycle{l}} for some $l \leq k$. Note that any such orientation can have at most $k$ sources. If the pattern has either $1$ or $2$ sources, then we can construct a \dagtree\ of width one by putting the sources in their own bags. If an orientation $\vec{H}$ has $l \geq 3$ sources then we can show that it is \reducible{\cycle{l}}:
		
		Let $s_0,...,s_{l-1}$ be the $l$ sources of $\vec{H}$, ordered in either clockwise or anti-clockwise. For $i \in [l]$, set $S_i = \{s_i\}$ and $I_i = \vec{H}(S_i) \cap \vec{H}(S_{i+1})$, note that these sets are always non-empty, as consecutive sources in a cycle intersect in a vertex. Additionally a source $s_i$ will only reach two intersection vertices, the one with $s_{i-1}$ and the one with $s_{i+1}$. Therefore, these sets will satisfy all the conditions in \Def{reducible}, giving that $\vec{H}$ is \reducible{\cycle{l}}.
	\end{proof}

	\subsection{Proof of \Lem{spasms}}
	
	We start by proving some auxiliary lemmas.
	
		\begin{lemma} \label{lem:merge}
		Given two directed patterns $\vec{P},\vec{P}'$ with $\dtw(\vec{P}) = \dtw(\vec{P}') = 1$, the graph $\vec{P}''$ obtained by either merging a vertex from $\vec{P}$ with a vertex from $\vec{P}'$ or two vertices $u,v$ from $\vec{P}$ connected by an edge with two vertices $u',v'$ from $\vec{P}'$ connected by an edge will also have $\dtw(\vec{P}'') = 1$.
	\end{lemma}
	
	\begin{proof}
		Consider a \dagtree{} $T$ of $\vec{P}$ and a \dagtree{} $T'$ of $\vec{P}'$ with $\dtw(T)=\dtw(T')=1$. Let $r \in T$ be any bag with $u \in \Reachable(r)$ (note that in the case of merging two vertices $v$ will also be in $\Reachable(r)$), similarly let $r' \in T'$ be any bag with $u' \in \Reachable(r')$. Now we consider the following scenarios:
		
		\begin{enumerate}[a)]
			\item $u$ is a source in $\vec{P}$ and $u'$ is a source in $\vec{P}'$: In this case we have that $r=u$ and $r'=u'$. Additionally, in $\vec{P}''$ the merged vertex $u$ will still be a source. Hence we can form a valid \dagtree{} $T''$ of $\vec{P}''$ with $\dtw=1$ by merging the bags $r$ and $r'$ into a single bag.
			
			\item $u$ is a source in $\vec{P}$ but $u'$ is not a source: In this case $r=u$, and in $\vec{P}''$ the merged vertex will not be a source, however $r'$ will still be one. We can create a valid \dagtree{} $T''$ of $\vec{P}''$ by replacing the bag containing $r$ in $T$ by the bag containing $r'$ and merging it with the $r'$ of $T'$. Note that all the original sources in $\vec{P}$ (except $u$) are still sources reaching the same vertices, while $r'$ now reaches the same vertices of $P$ that were reachable by $u$, hence all the reachability conditions will be satisfied.
			
			\item $u$ is not a source in $\vec{P}$ but $u'$ is a source in $\vec{P}'$: This case is analogous to the previous.
			
			\item Both $u$ and $u'$ were not sources in $\vec{P}$ and $\vec{P}'$ respectively: In this case we can construct a valid \dagtree{} $T''$ by connecting $r$ and $r'$ with an edge. In order for this to not be a valid \dagtree{}, we would need a source $s\neq r \in \vec{P}$ and a source $s' \neq r' \in \vec{P}'$ such that $\Reachable(s) \cap \Reachable(s') \not\subseteq \Reachable(r)$ or $\Reachable(s) \cap \Reachable(s') \not\subseteq \Reachable(r')$. Note that $\Reachable(s) \cap \Reachable(s')$ can at most contain only the merged vertex $u$ and all the vertices reachable by it, but both $r$ and $r'$ will also reach those vertices, hence the \dagtree{} is valid.
		\end{enumerate}
		\Fig{dag} shows the construction of the \dagtree{} in every case.
	\end{proof}
	
	\begin{figure}
		\centering
		\includegraphics[width=\textwidth*3/4]{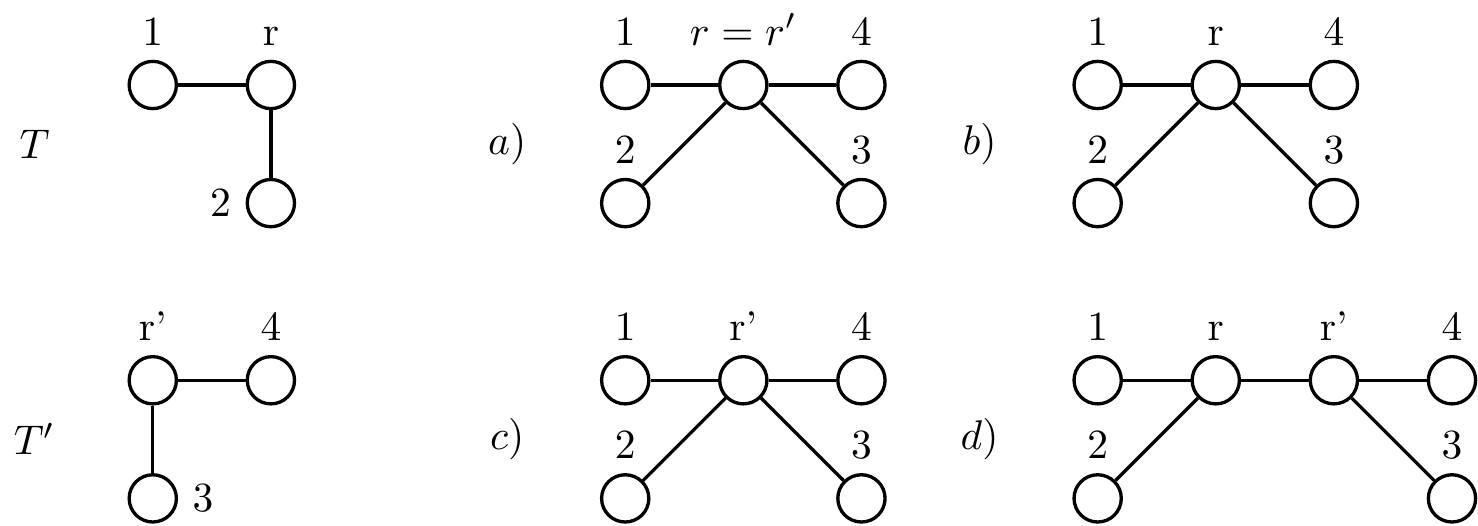}%
		\caption{The possible cases in \Lem{merge}. $T$ and $T'$ represent \dagtree{} of the graphs $\vec{P}$ and $\vec{P}'$, with $u \in \Reachable_{\vec{P}}(r)$ and $u' \in \Reachable_{\vec{P}'}(r')$.}
		\label{fig:dag}
	\end{figure}

	\begin{lemma} \label{lem:complexpattern}
		Given any \computable{\cycle{k}} undirected pattern $H$, and an undirected connected pattern $H'$ with $\LICL(H')<6$, the graphs resulting of the following operations are \computable{\cycle{k}}:
		\begin{itemize}
			\item Merge any vertex $u \in H$ with any vertex $u' \in H'$.
			\item Combine any edge $(u,v) \in E(H)$ with any edge $(u',v') \in E(H')$ by merging $u$ with $u'$ and $v$ with $v'$, and removing the duplicate edge.
		\end{itemize}
	\end{lemma}
	\begin{proof}
		
		Let $H''$ be the graph resulting of any of the operations. Consider any acyclic orientation $\vec{H}'' \in \Sigma(H'')$, let $\vec{H}$ and $\vec{H}'$ be the subgraphs of $\vec{H}''$ induced by the vertices of $H$ and $H'$ respectively, note that $\vec{H}$ and $\vec{H}'$ are themselves acyclic orientations of $H$ and $H'$ respectively. Because $\LICL(H') < 6$ we have that $\dtw(\vec{H}') = 1$.
		
		Now consider the pattern $\vec{H}$, we have two cases:
		\begin{itemize}
			\item $\dtw(\vec{H}) = 1$. In this case we can directly invoke \Lem{merge}, giving that the merged directed pattern $\vec{H}''$ also has $\dtw(\vec{H}') = 1$.
			
			\item $\vec{H}$ is \reducible{\cycle{l}}, for some $l \leq k$. Let $S_1,...,S_l$ be the sources of the $\cycle{l}$-reduction of $\vec{H}$. There must exist a set of sources $S_i$ that can reach the merged vertex/vertices. We define the following source sets for $\vec{H}''$: For $j\neq i$ we set $S''_j=S_j$ and for the remaining set of sources we set $S''_i = (S_i \cap S(\vec{H}'')) \cup S(\vec{H}')$. Note that if the merged vertex was a source in $\vec{H}$ it is possible that it is no longer a source in $\vec{H}''$. We set $I''_i = \vec{H}(S_{i})\cap \vec{H}(S_{i+1})$.
			
			If $s_i$ is still a source in $\vec{H}''$ then we set $s''_i = s_i$, otherwise there must be a source $s''$ that can reach the old $s_i$ in $\vec{H}''$. Such source will also reach all the vertices in $I_{i-1}$ and $I_{i}$. We can see that this is a valid $\cycle{l}$-reduction: The intersection will be valid as the only changes are that $I''_{i}$ and $I''_{i-1}$ might include additional vertices from $H'$. The source sets are also valid as all the sources will be part of exactly one source set. The other conditions in \Def{reducible} will follow from our construction.
		\end{itemize}
	\end{proof}
	
	Now we have all the tools to prove \Lem{spasms}. We break the result into different lemmas. The result for $\cycle{6}$ and $\cycle{7}$ follows directly from \Lem{sixseven}, as every pattern in those spasms will have at most $7$ vertices. We prove next the result for $\cycle{8}$ and $\cycle{9}$.
	
	\begin{lemma}
		All the patterns in $\Spasm(\cycle{8})$ and $\Spasm(\cycle{9})$ are \computable{\cycle{4}}.
	\end{lemma}
	
	\begin{figure}
		\centering
		\includegraphics[width=\textwidth*1/2]{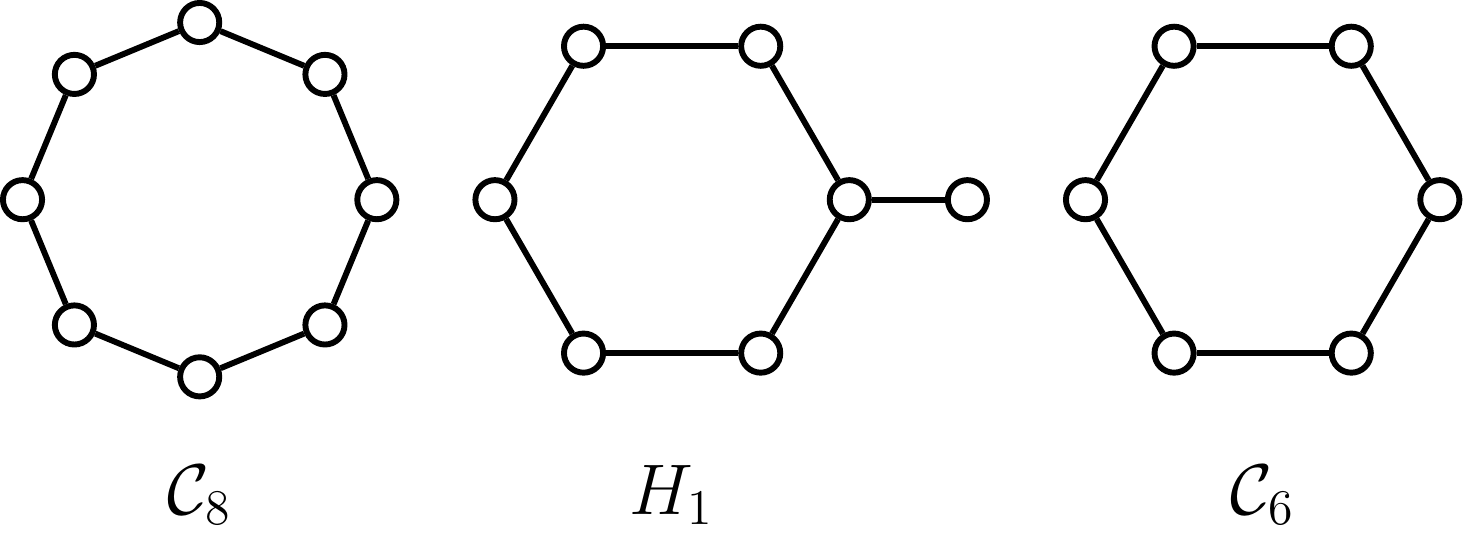}%
		\caption{The spasm of $\cycle{8}$, only including patterns with \LICL{} greater than $5$.}
		\label{fig:spasm8}
	\end{figure}
	
	\begin{figure}
		\centering
		\includegraphics[width=\textwidth*4/5]{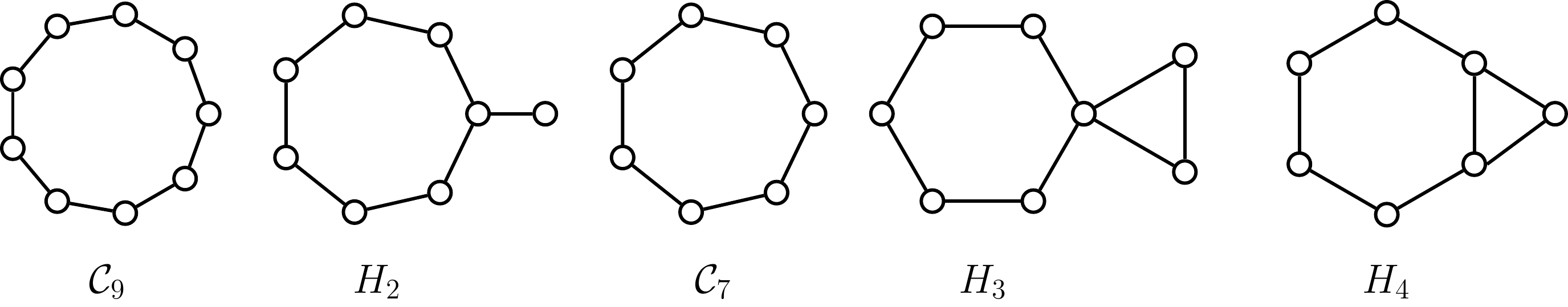}%
		\caption{The spasm of $\cycle{9}$, only including patterns with \LICL{} greater than $5$.}
		\label{fig:spasm9}
	\end{figure}
	
	\begin{proof}
		\Fig{spasm8} shows all the patterns in the spasm of $\cycle{8}$ with \LICL{} greater than $5$, similarly, \Fig{spasm9} shows all the patterns in the spasm of $\cycle{9}$ with \LICL{} greater than $5$. From \Lem{licl} we have that the rest of the patterns in the spasms will have a \dagtreewidth{} of $1$ in all their orientations. From \Lem{cycle_hom} we have that all the cycle patterns in both spasms will be $\cycle{4}$ or \computable{\cycle{3}}.
		
		Using \Lem{sixseven} we have that the patterns $H_1$ and $H_4$ are \computable{\cycle{3}}, as they only have $7$ vertices. For $H_2$ and $H_3$ we can use \Lem{complexpattern} to show that they are also \computable{\cycle{3}}, $H_2$ is formed by combining a $\cycle{7}$ with an edge, and $H_3$ is formed by combining a $\cycle{6}$ with a $\cycle{3}$.
	\end{proof}
	
	Finally, we prove the result for $\cycle{10}$.
	
	\begin{lemma}
		All the patterns in $\Spasm(\cycle{10})$ are \computable{\cycle{5}}.
	\end{lemma}
	\begin{proof}
		\begin{figure}
			\centering
			\includegraphics[width=\textwidth*4/5]{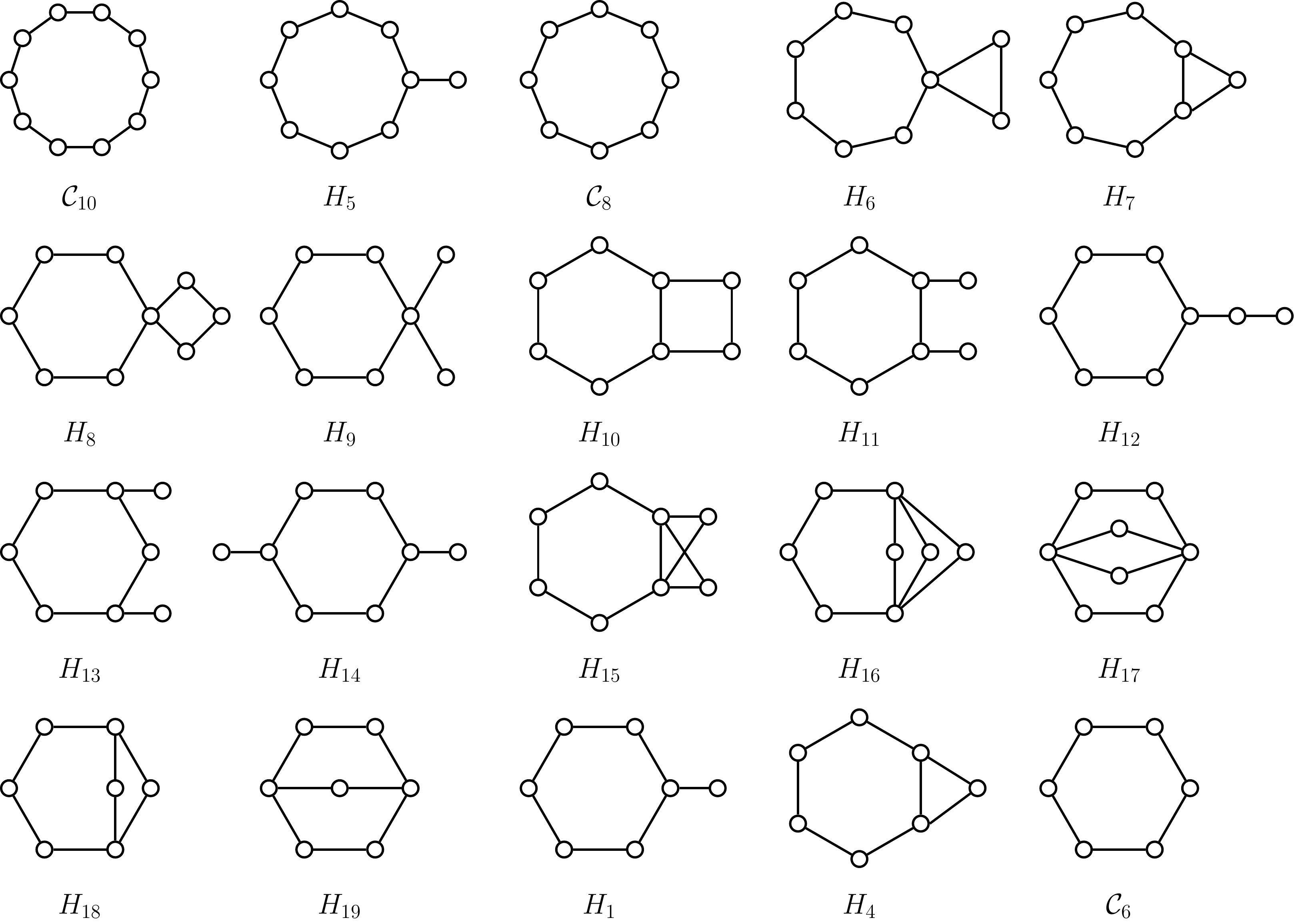}%
			\caption{The spasm of $\cycle{10}$, including only patterns with \LICL{} greater than $5$.}
			\label{fig:spasm10}
		\end{figure}
		
		\Fig{spasm10} shows all the patterns in the spasm of $\cycle{10}$ with $\LICL{}$ greater than $5$. We verify that all the patterns are \computable{\cycle{5}}:
		
		\begin{itemize}
			\item From \Lem{sixseven}, we have that patterns $H_1, H_4, H_{18}, H_{19}$ are \computable{\cycle{3}}, as they have $7$ vertices.
			\item From \Lem{cycle_hom} we have that $\cycle{10}, \cycle{8}, \cycle{6}$ are $\cycle{5}$, $\cycle{4}$ and \computable{\cycle{3}} respectively.
			\item From \Lem{complexpattern} we have that $H_5$ is \computable{\cycle{4}}, as it is obtained by combining a $\cycle{8}$ with a $2$-path.
			\item Again using \Lem{complexpattern} we can show that the following patterns are \computable{\cycle{3}}:
			\begin{itemize}
				\item $H_6$: $\cycle{7}$ + $\cycle{3}$ along a vertex.
				\item $H_7$: $\cycle{7}$ + $\cycle{3}$ along an edge.
				\item $H_8$: $\cycle{6}$ + $\cycle{4}$ along a vertex.
				\item $H_9$: $H_1$ + $2$-path along a vertex.
				\item $H_{10}$: $\cycle{6}$ + $\cycle{4}$ along an edge.
				\item $H_{11}$: $H_1$ + $2$-path along a vertex.
				\item $H_{12}$: $H_1$ + $2$-path along a vertex.
				\item $H_{13}$: $H_1$ + $2$-path along a vertex.
				\item $H_{14}$: $H_1$ + $2$-path along a vertex.
				\item $H_{15}$: $H_4$ + $\cycle{3}$ along an edge.
			\end{itemize}
		\end{itemize}
		
		Only left is to show that $H_{16}$ and $H_{17}$ are also \computable{\cycle{5}}:
		
		\begin{itemize}
			\item First, we look at $H_{17}$, consider any acyclic orientation of it, if it has $3$ sources or less then using \Lem{3sources} we know that the either $\dtw=1$ or the oriented pattern is \reducible{\cycle{3}}. Hence, we consider the orientations with $4$ or more sources, the only option to have $4$ sources is by having both central vertices as sources, and then one of the two top vertices as source and one of the two bottom vertices also as source. One can verify that all the resultant orientations will have a $\dtw=1$, this can be seen as the two central sources will reach the same vertices and the top and bottom source can only intersect in vertices also reachable by the central ones.
			
			\item Now we look at $H_{16}$, consider any acyclic orientation. If it has $3$ sources again we can use \Lem{3sources}. Otherwise if it has $5$ sources it will have at most $3$ intersection vertices, which by \Lem{3intersections} indicated that the pattern has $\dtw=1$ or it is \reducible{\cycle{3}}. Now consider the orientations with $4$ sources, it must happen that at least $2$ of the sources are in the right side (the triple $2$-path), but such sources will reach the same intersection vertices. Meaning that they can be put in the same set of sources, obtaining a \reducible{\cycle{3}} pattern.
		\end{itemize}
	\end{proof}
	
	\subsection{Proof of \Lem{uppercycle}}
	\uppercycle*
	\begin{proof}
		From \Lem{spasms} we have that for all $6 \leq k \leq 10$ all patterns in $\Spasm(\cycle{k})$ are \computable{\cycle{\lfloor k/2 \rfloor}}. From \Lem{cycle_complexity} we have that we can compute $\WSubNI{\cycle{k}}$ in time $O(m^{d_k})$. Hence we can use \Lem{computable} to compute $\Hom{G}{H}$ in time $f(\degen)O(n^{d_{\lfloor k/2 \rfloor}})$ for all patterns in the spasm of $\cycle{k}$, then we can obtain $\Sub{G}{H}$ using inclusion-exclusion.
	\end{proof}
	
	\section{Proofs of Section \ref{sec:others}}
	
	\subsection{Proof of \Lem{pk_computable}}
	
	For reference we restate the definition of $\mathcal{P}_{i,s}$ and $\mathcal{P}_k$.
	
	\calp*
	
	We can show that if we have a graph $P$ satisfying the first three conditions definition, then every \reducible{P} pattern can also be reduced to some graph in $\mathcal{P}_{i,s}$, for some $i$ and $s$.
	
	\begin{lemma} \label{lem:simplify}
		Let $P$ be a hypergraph with $i\geq4$ vertices and $s \geq 4$ hyperedges such that:
		\begin{enumerate}
			\item Every vertex has degree at least $2$.
			\item Every hyperedge contains at least $2$ vertices.
			\item No hyperedge is a subset of any other hyperedge.
		\end{enumerate}
		Then, for some $i' \leq i$, there exists a hypergraph $P' \in \mathcal{P}_{i',s}$ such that every directed pattern $\vec{H}$ that is \reducible{P} is also \reducible{P'}.
	\end{lemma}
	\begin{proof}
		Let $E(u)$ be the set of hyperedges of $P$ that contain the vertex $u$, and let $R$ be the sum of arities of all hyperedges in $P$.
		
		If $P$ already satisfies the fourth condition of \Def{group} then $P \in \mathcal{P}_{i,s}$ and we are done.  Otherwise, there is a pair $u,v \in V(P')$ such that $E(u) \subseteq E(v)$. We show that in that case we can construct a graph $P'$ with $|V(P')| \leq i$ and $E(P') = s$ that still satisfies the initial three constraints and has sum of arities $R' < R$, and for which every pattern that is \reducible{P} is also \reducible{P'}:
		
		We have $E(u) \subseteq E(v)$, we distinguish two cases:
		\begin{itemize}
			\item If $E(u) = E(v)$ then we can just merge the vertices $u$ and $v$ into a single vertex $uv$, every hyperedge that contained both $u$ and $v$ will just contain $uv$ instead. $P'$ has the same number of edges than $P$ but one vertex less, the total arity $R$ of $P'$ will also be less than the one of $P$. Only left is to show that every \reducible{P} pattern $\vec{H}$ is also \reducible{P'}. 
			
			Let $\{S_e : e \in E(P)\}$ and $\{I_v : v \in V(P)\}$ be the sets of sources and intersection vertices of $\vec{H}$ that achieve $P$-reducibility, setting $I_{uv} = I_u \cup I_v$ and removing $I_u$ and $I_v$ while maintaining all the other intersection and source sets the same will achieve $P'$-reducibility: for any of the vertices in $I(\vec{H})$ the connectivity condition from \Def{reducible} will not be affected as $uv$ will neighbor all the vertices that $u$ and $v$ already neighbor in $P$. Our construction guarantees that the two final conditions from \Def{reducible} are satisfied.
			
			\item If $E(u) \subset E(v)$ then we have that every edge containing $u$ also contains $v$. We construct $P'$ by adding the vertex $uv$ and removing $v$, for every hyperedge that contained both $u$ and $v$ we instead only contain $uv$, except for a single hyperedge $e$ that will contain both $u$ and $uv$. $P'$ will have the same number of vertices and hyperedges than $P$ but the total arity will reduce by at least $1$. We show that every \reducible{P} pattern $\vec{H}$ is also \reducible{P'}. 
			
			Let $\{S_e : e \in E(P)\}$ and $\{I_v : v \in V(P)\}$ be the sets of sources and intersection vertices of $\vec{H}$ that achieve $P$-reducibility. We set $I_{uv} = I_u \cup I_v$ and remove the set $I_v$ while maintaining all the other sets the same. We can show that the new sets will achieve $P'$-reducibility: the edge $e$ that connects both $uv$ and $v$ guarantees that the intersection vertices in $I_u$ and $I_v$ will still satisfy the connectivity condition from \Def{reducible}. Our construction guarantees that the two final conditions from \Def{reducible} are also satisfied.
		\end{itemize}
		
		If $P'$ satisfies the last condition from \Def{group} then it must be the case that $P' \in \mathcal{P}_{i',s}$ for some $i' \leq i$. Otherwise another pair of vertices must be violating the fourth condition, and we can then repeat the previous process. Eventually we are guaranteed to reach a graph $P'$ which satisfies all the conditions and it is in $\mathcal{P}_{i',s}$ for some $i' \leq i$. This is guaranteed because the value of $R$ will decrease in every iteration, but it can not be lower than $2s$ (twice the number of hyperedges).
	\end{proof}
	
	We can now prove \Lem{pk_computable}.
	
	\pkcomputable*
	\begin{proof}
		Let $\vec{H}$ be any directed acyclic graph with $k$ vertices. If $\vec{H}$ admits a \dagtree{} with $\tau=1$ then we are done, if $\vec{H}$ has at most $3$ sources or $3$ intersection vertices then $\vec{H}$ is \reducible{\cycle{3}} and we are also done.
		
		Otherwise, $\vec{H}$ has at least $4$ sources and $4$ intersection vertices. Let $S = S(\vec{H})$ be the set of sources of $\vec{H}$. Initialize sets $S_e = \{s_e\}$ for every source $s_e \in S$, we also construct trees of sources $\cT_e$ for each source $s_e$ containing a single node $s_e$. Let $I^*$ be the set of intersection vertices reachable by two distinct sets $S_e,S_e'$. For every set $S_e$, let $I(e)$ be the set of intersection vertices in $I^*$ reachable by $S_u$. 
		
		If there exists $e \neq e' \in S$ such that $I(e) \subseteq I(e')$, then set $S_{e'} = S_e \cup S_{e'}$ and delete $S_e$, combine the trees $\cT_e$ and $\cT_{e'}$ by adding an edge between $s_e$ and $s_{e'}$, and update the set $I^*$. Repeat until we obtain $s \leq |S|$ sets of sources such that there are no $e \neq e'$ such that $I(e) \subseteq I_{e'}$. Let $S^*$ be the collection of the remaining sets. For every set $S_e \in S^*$, the graph $\vec{H}(S_e)$ has a \dagtreewidth{} of $1$, as $\cT_e$ is a valid \dagtree{} decomposition of $\vec{H}(S_e)$.
		
		Consider the hypergraph $P$ where every intersection vertex in $I^*$ is a vertex, and every source set $S_e$ is a hyperedge $e$ containing the vertices in $I(e)$.
		
		We can show that $\vec{H}$ is \reducible{P}. For each edge $e \in E(P)$ we use $S_e$ as the corresponding source of sets, and for each vertex $v \in V(P)$ corresponding to the intersection vertex $v' \in I^*$ we set $I_v = \{v'\}$. We can verify the conditions from \Def{reducible}:
		
		\begin{itemize}
			\item For every vertex $v \in V(P)$ and every edges $e \ni v$ we have that $v' \in I(e)$, which implies that $I_v \subseteq I(e)$.
			\item For every hyperedge $e \in V(P)$ we have that $e$ contains all the vertices $v$ for which $v' \in I(e)$. Hence $\bigcup_{v\in e}I_v = I(e)$.
		\end{itemize}
		
		We show that $P$ must satisfy the first three conditions from \Def{group}:
		\begin{enumerate}
			\item By construction, every vertex in $I^*$ must be reachable by at least two distinct sets, hence the corresponding vertex in $P$ will have a degree of at least $2$.
			\item Every set $S_e$ must reach at least two different intersection vertices in $I^*$, otherwise there we could combine it with a different set of sources, with other set of sources that reaches the same intersection vertex.
			\item A hyperedge $e$ being a subset of another hyperedge $e'$ would imply $I(e) \subseteq I_{e'}$, which is not allowed by the construction.
		\end{enumerate}
		
		Let $i$ be the number of vertices in $P$ and $s$ the number of edges. Note that $s \geq 3$, otherwise the original pattern $\vec{H}$ has \dagtreewidth{} of $1$, this implies $i \geq 3$, or the construction would violate one of the previous conditions. If $s=3$ then $i=3$ and the $P = \cycle{3}$ which means $\vec{H}$ is \reducible{\cycle{3}}. 
		
		Otherwise we have that $i,s \geq 4$ and $i+s \leq k$. Using \Lem{simplify} we have that for some $i' \leq i$ there is a pattern $P'\in \mathcal{P}_{i',s}$ such that $\vec{H}$ will be \reducible{P'}, and we have that  $\mathcal{P}_{i',s} \subseteq \mathcal{P}_k$.
	\end{proof}
	
	\subsection{Proof of \Lem{nine_content}}
	
	\begin{figure}
		\centering
		\includegraphics[width=\textwidth*1/4]{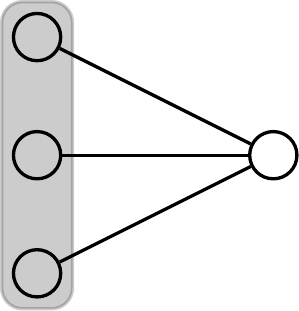}%
		\caption{The hypergraph $\hyperone$. Formed by one hyperedge of arity $3$ (in gray) and three normal edges.}
		\label{fig:hyperone}
	\end{figure}
	
	From the definition of $\cP_9$ we have that:
	\[
		\cP_9 = \cP_8 \cup \cP_{4,5}  \cup \cP_{5,4} = \{\cycle{3}\} \cup \cP_{4,4} \cup \cP_{4,5}  \cup \cP_{5,4}
	\]
	We will show which patterns form each of these sets. First, we can show that $\cP_{4,4}$ is formed by the $4$-cycle, the $3$-simplex and $\hyperone$ (\Fig{hyperone}).
	
	\begin{claim}
		$\mathcal{P}_{4,4} = \{\cycle{4}, \simplex{3}, \hyperone \}$
	\end{claim}
	\begin{proof}
		Let $P$ be a hypergraph in $\mathcal{P}_{4,4}$. Note that $P$ can not have any hyperedge of arity $4$, otherwise all the other hyperedges would be subsets of it, breaking condition $3$ in \Def{group}. Hence $P$ may contain only hyperedges with arity $2$ and $3$, consider the following cases:
		\begin{itemize}
			\item $P$ has $4$ hyperedges of arity $2$: We can show that in this case $P$ must be the $4$-cycle. If $P$ were to contain a triangle then we have that the vertex not in the triangle will have a degree of at most $1$, this contradicts condition $1$ in \Def{group} and hence $P$ must not contain a triangle. The only graph with $4$ vertices and $4$ edges that does not contain a triangle is the $4$-cycle.
			
			\item $P$ has $3$ hyperedges of arity $2$ and $1$ of arity $3$: Let $e$ be the hyperedge of arity $3$ and $v$ the vertex not in $e$. Every vertex in $e$ must be part of an additional edge in order to have degree $2$, the only vertex they can connect without violating condition $3$ in \Def{group} is $v$, hence the three edges with arity $2$ will connect $v$ with each of the three vertices in $e$, giving the pattern $\hyperone$ seen in \Fig{hyperone}.
			
			\item $P$ has $2$ hyperedges of arity $2$ and $2$ of arity $3$: The two hyperedges of arity $3$ must intersect in two vertices, let $u,u'$ be those vertices and $v,v'$ the two remaining vertices. The remaining two edges of arity $2$ can not contain $u$ or $u'$ as they would be subsets of one of the hyperedges of arity $3$. Hence only the other two vertices $v,v'$ can be part of an edge. But this only gives one possible edge, hence we can not have a graph $P \in \mathcal{P}_{4,4}$ with this configuration.
			
			\item $P$ has $1$ hyperedges of arity $2$ and $3$ of arity $3$: The two hyperedges of arity $3$ must intersect in two vertices, let $u,u'$ be those vertices and $v,v'$ the two remaining vertices. The third hyperedge of arity $3$ can not contain both $u$ and $u'$ so it will contain $v,v'$ and either $u$ or $u'$. But this means that every pair of vertices is already part of some hyperedge, and we can not have an extra edge of arity $2$ without violating the condition $3$ in \Def{group}.
			
			\item $P$ has $4$ hyperedges of arity $3$: There are $4$ possible hyperedges of arity $3$ in a graph with $4$ vertices, hence $P$ must contain all of them, the result is the simplex of arity $3$, $\simplex{3}$.
		\end{itemize}
	\end{proof}

	$\cP_{4,5}$ only contains the diamond graph $\diamondgraph$, also known as $K_4$ minus one edge.

	\begin{claim}
		$\mathcal{P}_{4,5} = \{\diamondgraph \}$.
	\end{claim}
	\begin{proof}
		We first can show that any graph $P \in \mathcal{P}_{5,4}$ may not contain any hyperedge of arity greater than $2$: If it were to contain a hyperedge of arity $4$ then all the other hyperedges would be subset of it, breaking condition $3$ in \Def{group}. Also $P$ can not contain $5$ hyperedges of arity $3$ as there are at most $4$ such hyperedges with $4$ vertices. Hence $P$ must contain at least an edge of arity $2$. We consider the rest of cases:
		\begin{itemize}
			\item If $P$ contains $3$ or $4$ hyperedges of arity $3$ then every pair of vertices in $P$ will be together in one of the hyperedges and we can not add any edges of arity $2$. Hence this can not happen.
			\item If $P$ contains $2$ hyperedges of arity $2$, then there is only one pair of vertices that is not together in some hyperedge, but need to add $3$ edges.
			\item If $P$ contains only $1$ hyperedge of arity $3$, we can add an edge connecting every vertex in the hyperedge with the free vertex, but we will still be missing one edge.
		\end{itemize}
		Therefore, $P$ can only contain edges of arity $2$. The only pattern with $4$ vertices and $5$ edges is the diamond or $K_4$ minus one edge.
	\end{proof}

	\begin{figure}
		\centering
		\includegraphics[width=\textwidth*1/4]{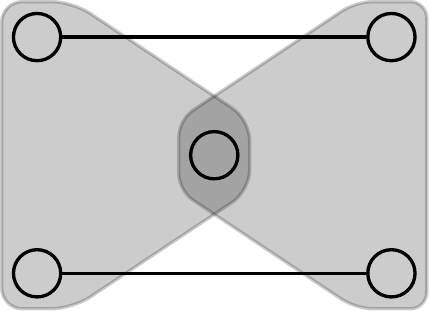}%
		\caption{The hypergraph $\hypertwo$. Formed by two hyperedges of arity $3$ (in gray) and two normal edges.}
		\label{fig:hypertwo}
	\end{figure}
	
	Finally, $\mathcal{P}_{5,4}$ only contains the hypergraph $\hypertwo$ (\Fig{hypertwo}).
	
	\begin{claim}
		$\mathcal{P}_{5,4} = \{\hypertwo\}$.
	\end{claim}
	\begin{proof}
		Let $P$ be a hypergraph with $5$ vertices and $4$ hyperedges that satisfies all the conditions in \Def{group}. We prove that $P = \hypertwo$. $P$ can not contain a hyperedge with arity $5$, as then the other hyperedges would necessarily be subsets of it. 
		
		Suppose $P$ contains an arity $4$ hyperedge, let $e$ be such hyperedge, every vertex in $e$ must be part of at least another hyperedge to have degree of at least $2$, also, for every pair of vertices $u,u'\in e$ the set of edges containing $u$ can not be a subset or equal to the set of edges containing $u'$. However we only have $3$ remaining hyperedges (lets call them $\{e_1,e_2,e_3\}$), we can create at most $3$ subsets of the remaining hyperedges that will not be subset of each other, either $\{e_1\},\{e_2\},\{e_3\}$ or $\{e_1,e_2\}$,$\{e_1,e_3\}$,$\{e_2,e_3\}$, we can assign each subset to one of the vertices but the forth vertex will necessarily create a conflict either by being contained by a subset or a superset of those groups of hyperedges.
		
		Hence, $P$ can not contain an arity $4$ hyperedge. Neither can $P$ contain only arity $2$ hyperedges, as in that case the average degree (and hence the minimum degree) of $P$ would be below $2$. Hence $P$ must be formed by arity $2$ and $3$ hyperedges. Similarly if $P$ contains $1$ arity $3$ hyperedge and $3$ arity $2$ edges we will have the average degree is below $2$.
		We show that $P$ can not contain two arity $3$ hyperedges that intersect in more than one vertex:
		
		Consider otherwise, there are two hyperedges $e,e'$ intersecting in $u,u'$. Let $e = \{u,u',v\}$ and $e'=\{u,u',v'\}$. We have two more hyperedges. Both hyperedges must include the additional vertex $w$ in order for it to have a degree of $2$. We need $u$ and $u'$ to be part of at least some hyperedge without the other to avoid breaking the fourth condition of \Def{group}, hence one of the remaining hyperedges must contain $\{u,w\}$ and the other $\{u,w\}$. However, to avoid breaking the condition we need both $v$ and $v'$ to be in some hyperedge without $u$ and without $u'$, we can achieve this for $v$ by adding $v$ to the previous hyperedges, but then we will not have any additional hyperedge connecting $v'$, as the maximum arity is $3$. Therefore, $P$ can not contain such hyperedges.
		
		If $P$ contains $3$ or more arity-$3$ hyperedges then at least two of them will intersect in more than one vertex. Thus, $P$ must contain $2$ hyperedges with arity $2$ and two more with arity $3$. The only hypergraph of that form that satisfies all constraints is $\hypertwo$.
	\end{proof}

	Combining the three previous claims gives \Lem{nine_content}.
	
	\subsection{Proof of \Lem{nine_star}}
	
	Let $\simplex{r}$ be the $r$-simplex. We can show that every pattern that is \reducible{\simplex{r}} is also \reducible{\cycle{r+1}}. See \Fig{simplex} for an example.
	\begin{lemma} \label{lem:simplex}
		For all $r\geq 3$, every \reducible{\simplex{r}} pattern is \reducible{\cycle{r+1}}.
	\end{lemma}
	\begin{proof}
		In this proofs all the indexes are given in module $r+1$.
		
		Let $\vec{H}$ be a \reducible{\simplex{r}} pattern, for each vertex $v \in V(P)$ let $I_s$ be the sets of intersection vertices that achieve the $\simplex{r}$-reducibility. Similarly for each $e \in E(P)$ let $S_e$ be the set of sources that achieve the $\simplex{r}$-reducibility.
		
		Arbitrary sort the $r+1$ vertices of $\simplex{r}$ $v_0,...,v_{r}$, and let $e_i =\{V(\simplex{r})\setminus \{v_i\}\}$. We define $I(e)$ as the set of intersection vertices reached by the set of sources $S_e$ in $\vec{H}$, hence we will have $I(e_i) = \bigcup_{j\neq i} I_j$.
		
		We define $r+1$ new intersection sets as follows, $I'_i = \bigcup_{j=i+1}^{r-1+i} I_i$ for $i \in [0,4]$. Note that for all $i$, $I'_i \cup I'_{i+1} = I(e_i)$ and both $I'_i \subseteq I(e_i)$ and $I'_i \subseteq I(e_{i-1})$. Hence we can use the original source sets $S_e$ with the new intersection sets $I'_v$ to achieve $\cycle{r+1}$-reducibility, assigning $I'_i$ to the $i$-th vertex of the cycle and $S_{e_i}$ to the edge connecting the $i$ and $i+1$-th vertices.
	\end{proof}
	
	\begin{figure}
		\centering
		\includegraphics[width=\textwidth*3/4]{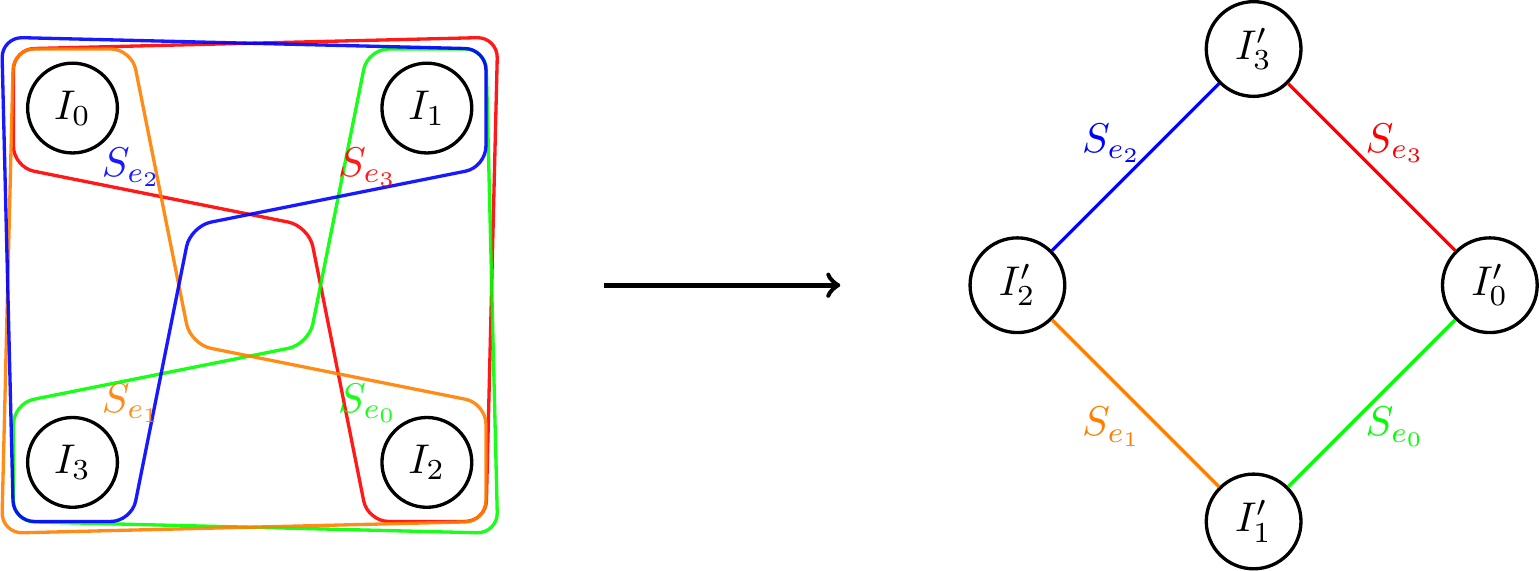}%
		\caption{An example of apply \Lem{simplex} to the $3$-simplex, which becomes a $4$-cycle. Here $I'_i =  I_{i+1}\cup I_{i+2}$.}
		\label{fig:simplex}
	\end{figure}

	We can now prove \Lem{nine_star}.
	
	\ninestar*
	\begin{proof}
		Let $H$ be a \computable{\cP_9} pattern. Consider any of its orientations $\vec{H}$:
		\begin{itemize}
			\item If $\vec{H}$ is \reducible{P} for some $P \in \cP_9 \setminus \{\simplex{3}\}$ then we have that $P \in \cP^*_9$.
			\item Otherwise, $\vec{H}$ is \reducible{\simplex{3}}, but in that case using \Lem{simplex} we have that $\vec{H}$ is \reducible{\cycle{4}}, which is in $\cP^*_9$.
		\end{itemize}
		Hence, all the orientations of $H$ reduce to some pattern in $\cP^*_9$, and hence $H$ will be \computable{\cP^*_9}.
	\end{proof}
	
	\subsection{Proof of \Lem{allnine}}
	
	From \Lem{cycle_complexity} we have for $\cycle{3}$ and $\cycle{4}$, we can compute Col-WSub in time $O(m^{d_3})$ and $O(m^{d_4})$ respectively. We verify that for the remaining three graphs in $\cP^*_9$ ($\diamondgraph, \hyperone, \hypertwo)$ we can compute Col-WSub in $\tilde{O}(m^{5/3})$ time.
	
	\begin{lemma}
		There is an algorithm that for any colored, weighted input graph $G$ with $m$ edges computes $\WSub{G}{\hyperone}$ in $O(m^{5/3})$ time.
	\end{lemma}
	\begin{proof}
		Let $v$ be the vertex of $\hyperone$ that is not part of the arity-$3$ hyperedge. Set $\Delta = m^{1/3}$. We can split the vertices of $V^{(v)}(G)$ into two groups, depending on their degree: let $V^H = \{v \in V^{(v)}(G) : d(v) > \Delta\}$ and $V^L = \{v \in V^{(v)}(G) : d(v) \leq \Delta\}$. Every copy of $\hyperone$ in $G$ must contain a vertex in either $V^L$ or $V^H$. We show how to list all copies of $P$:
		
		\begin{itemize}
			\item For every vertex $u \in V^H$ we can iterate over every hyperedge of arity $3$ in $G$ and verify if $u$ connects with each of the vertices in the hyperedge and they have the right colors. There are at most $O(\frac{m}{\Delta})$ vertices in $V^H$. Thus this step will take $O(m\frac{m}\Delta) = O(m^{5/3})$ total time.
			\item For every vertex $u \in V^L$ we can list all triplet of edges connecting $u$ to vertices in each of the other three layers, forming a 3-star and then check if there is a hyperedge that contains the three endpoints of the star. There will be at most $\Delta$ such edges, hence we can list all combinations in time $O(m \Delta^2) = O(m^{5/3})$.
		\end{itemize}
		For each copy of $\hyperone$ that we find we can just multiply the edge/hyperedge weights and aggregate the products to obtain the value of $\WSub{G}{\hyperone}$.
	\end{proof}

	\begin{lemma}
		There is an algorithm that for any colored, weighted input graph $G$ with $m$ edges computes $\WSub{G}{\diamondgraph}$ in $\tilde{O}(m^{3/2})$ time.
	\end{lemma}
	\begin{proof}
		Let $e$ be the diagonal edge in $\diamondgraph$, let $u_1$ and $u_2$ be the endpoints and $u_3,u_4$ the other two vertices in the diamond.
		
		We can enumerate all the copies of the colored triangles $u_1-u_2-u_3$ and $u_1-u_2-u_4$ in $G$ in time $O(m^{3/2})$. We create a hashmap $d_{u_3}$ and for each triangle of the first type with vertices $v_1-v_2-v_3$ we add the product of $w(v_1,v_3)$ and $w(v_2,v_3)$ to $d_{u_3}$. Similarly we will fill a hashmap $d_{u_4}$ with the triangles of the second type. Then for each edge $(v_1,v_2)$ connecting vertices of the type $u_1$ and $u_2$ we will compute $w(v_1,v_2) \cdot d_{u_3}((v_1,v_2)) \cdot d_{u_4}((v_1,v_2))$ and aggregate all the results to obtain $\WSub{G}{\diamondgraph}$. The total runtime of the algorithm will be $\tilde{O}(m^{3/2})$.
	\end{proof}
	
	\begin{lemma} \label{lem:hypertwo}
		There is an algorithm that for any colored, weighted input graph $G$ with $m$ edges computes $\WSub{G}{\hypertwo}$ in $\tilde{O}(m^{5/3})$ time.
	\end{lemma}
	\begin{proof}
		To compute $\WSub{G}{\hypertwo}$ we can perform thresholding on the vertices of $G$, dividing them in two groups for every color depending on their degrees. Set $\Delta = m^{1/3}$, and let $u_i$ be a vertex $\hypertwo$ we define $V^{(u_i)}_H = \{v \in V^{(u_i)(G)}: d(v) > \Delta\}$ and $V^{(u_i)}_L = \{v \in V^{(u_i)(G)}: d(v) \leq \Delta\}$. Every appearance of $\hypertwo$ in $G$ must be as one of the $32$ different combinations of vertices from the high and low degree sets. We show how to compute the counts for each combinations. Each of the combinations will correspond to at least one of the $4$ following cases, which are shown in \Fig{hypertwocomp}. We use $u_c$ for the central vertex and $u_1,u_2,u_3,u_4$ for the exterior vertices:
		\begin{enumerate}
			\item The central vertex is in the low degree set: In this case we can list all pairs of hyperedges for each vertex in $V^{(u_c)}_L$, and then check if they induce a valid copy of $\hypertwo$ in constant time, in that case we just aggregate the product of weights of the hyperedges in the copy. There will at most $ O(m\Delta) = O(m^{4/3})$ pairs of hyperedges to check.
			
			\item All the vertices have high degree: We split the graph in two parts and list each of the parts separately. Each part will be formed by one hyperedge of arity $3$ and one of the normal edges. Both parts intersect in three vertices forming a diagonal. Let $u_c,u_2,u_4$ be the vertices forming the intersection.
			
			We can list every copy of each of the halfs. For each half we list all possible hyperedges corresponding to the arity-$3$ hyperedge and vertices in the high set of the vertex that reach the normal edge. There will be at most $O(m\frac{m}{\Delta}) = O(m^{5/3})$ pairs for each of the halfs. We aggregate the products of weights of each copy into two hashmaps (one for each half) using the values of $u_c,u_2,u_4$ as the key.
			
			We then iterate over the hashmaps multiplying the values with similar entries and aggregating the counts. This will take total $\tilde{O}(m^{5/3})$ time. 
			
			\item A hyperedge contains $1$ low degree vertex: Let $e$ be the hyperedge with the low degree vertex and $e'$ the other hyperedge. Let $u_1$ be the low degree vertex and $u_4$ the high degree. We split $\hypertwo$ in two halfs, one containing $e$ and the edge $(u_1,u_2)$ and the other containing $e'$ and the edge $(u_3,u_4)$, note that they intersect at the vertices $u_c,u_2,u_4$. 
			
			We can list every copy of each of the halfs. For the half containing the low degree vertex we can just iterate over the vertices in $V^{(u_1)}_L$ and look at each pair of neighbors, there will be at most $O(m\Delta) = O(m^{4/3})$ possible pairs. For the other half we can list all possible hyperedges corresponding to $e'$ and vertices in $V^{(u_4)}_H$. There will be at most $O(m\frac{m}{\Delta}) = O(m^{5/3})$ pairs. For each of the halfs we aggregate the products of weights of each copy into two hashmaps (one for each half) using the values of $u_c,u_2,u_4$ as the key. 
			
			We then iterate over the hashmaps multiplying the values with similar entries and aggregating the counts. This will take total $\tilde{O}(m^{5/3})$ time.
			
			\item A hyperedge contains $2$ low degree vertices: Let $e$ be the vertex of $\hypertwo$ with the two low degree vertices. We can iterate over all the hyperedges in $G$ with the same colors containing two low degree vertices, and then iterate over each pair of edges coming from each of the low degree vertices. We can then check in constant time if the hyperedge together with the two edges induce a copy of $\hypertwo$ in $G$, in which case we add the product of the weight to the total count. There will be at most $O(m\Delta^2)= O(m^{5/3})$ subgraphs to verify.
		\end{enumerate}
	\end{proof}
	
	\begin{figure}
		\centering
		\includegraphics[width=\textwidth]{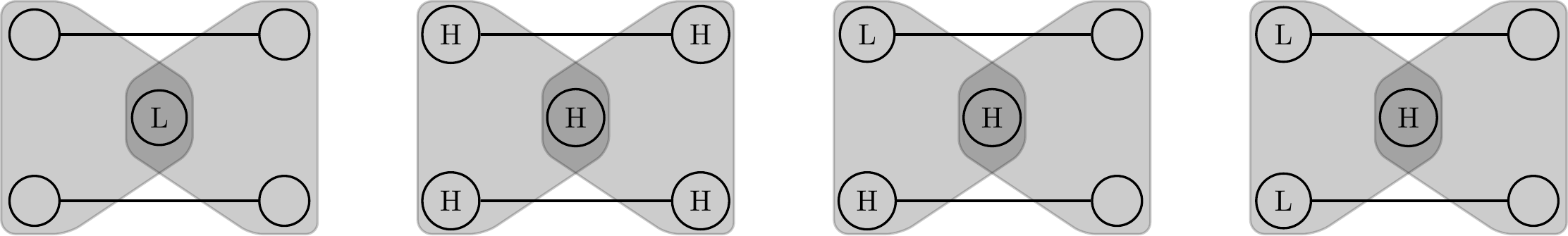}%
		\caption{The possible cases for thresholding of hypergraph $\hypertwo$ that are analyzed in \Lem{hypertwo}.}
		\label{fig:hypertwocomp}
	\end{figure}

	\section{Proof of \Lem{wsub_cycles}}
	
	The proof will closely follow the structure of Sections $4$ and $5$ in \cite{GiLeSh+23}, with the necessary modifications, as we are simply adapting the existing algorithms.
	
	\subsection{Combinatorial algorithm for cycles}
	
	Let $P_r$ be the path with $r+1$ vertices and $r$ edges with consecutive vertices having consecutive colors. For a colored graph $G$ we use $V^{(l)}(G)$ to denote the set of vertices with color $l$ or just $V^{(l)}$ if $G$ is clear from the context. In this section indexes and colors are taken modulo $k$.
	
	We first show the following auxiliary lemma, equivalent to Lemma $4.1$ in \cite{GiLeSh+23}:
	
	\begin{lemma} \label{lem:aux_cycle}
		Let $G$ be a weighted colored graph with $m$ edges. There is an algorithm that for integers $r, o$ and $\Delta \geq 1$ computes:
		\begin{itemize}
			\item For every pair $u,v \in V(G)$, the total weight $N_{r,o,\Delta}(u,v)$ of paths $P_r =\{p_0,...,p_{r}\}$ in $G$ such that $p_i \in V^{(i+o)}(G)$ for all $i$, $p_0 = u$, $p_{r} = v$ and $d(p_i) \leq \Delta$ for every $0 < i < r$.
			\item For every pair $u,v \in V(G)$ with $d(u) > \Delta$ and for every function $f: \{1,...,r\} \to \{\text{Low},\text{High}\}$ the total weight $M_{r,o,\Delta,f}$ of paths $P_r =\{p_0,...,p_r\}$ in $G$ such that $p_i \in V^{(i+o)}(G)$ for all $i$, $p_0 = u$, $p_{r} = v$ and for all $i>0$, $d(p_i) \leq \Delta$ if $f(i) =\text{Low}$ and $d(p_i) > \Delta$ if $f(i) =\text{High}$.
			\item For every pair $u,v \in V(G)$ with $d(v) > \Delta$ and for every function $f: \{0,...,r-1\} \to \{\text{Low},\text{High}\}$ the total weight $M'_{r,o,\Delta,f}$ of paths $P_r =\{p_0,...,p_r\}$ in $G$ such that $p_i \in V^{(i+o)}(G)$ for all $i$, $p_0 = u$, $p_{r} = v$ and for all $i<r$, $d(p_i) \leq \Delta$ if $f(i) =\text{Low}$ and $d(p_i) > \Delta$ if $f(i) =\text{High}$.
		\end{itemize}
		The first can be computed in time $\tilde{O}(m\Delta^{r-1})$ and the second and third in time $\tilde{O}(m^2/\Delta)$. There are at most $O(m\Delta^{r-1})$ pairs $u,v$ for which $N_{r,\Delta}(u,v) > 0$.
	\end{lemma}
	\begin{proof}
		For the first quantity, look at the all the edges between $V^{(o)}(G)$ and $V^{(1+o)}(G)$ and fix the position of the first two vertices $p_0$ and $p_1$, there will be at most $O(m)$ choices. We have $d(p_1) \leq \Delta$, hence we have at most $\Delta$ choices for $p_2$. Similarly, all vertices $p_2,...,p_{r-1}$ will have low degree and hence given $p_i$ we will have at most $\Delta$ choices for $p_{i+1}$, we need to do this $r-1$ times. Therefore, there will be at most $O(m\Delta^{r-1})$ possible paths. For each path we multiple the weight of the edges and aggregate the product into the hashmap $N_{r,o,\Delta}$ using the endpoints of the path as key. There will be at most $O(m\Delta^{r-1})$ possible non-zero entries.
	
		For the second quantity, we do induction in $r$. For $r=1$ we just need to looks at all the edges with one endpoint $u$ in $V^{(o)}(G)$ with $d(u)>\Delta$ and another endpoint $v$ in $V^{(1+o)}(G)$. We then store the weight of the edge $(u,v)$ in the corresponding $M_{r,o,\Delta,f}(u,v)$ with $f$ depending on the degree of $v$. There are at most $O(m/\Delta)$ vertices with high degree and they can have at most $n$ neighbors, hence we can complete this step in $O(nm/\Delta) = \tilde{O}(m^2/\Delta)$ time.
		
		For the inductive step, assume that we can compute all the counts for $P_{r-1}$ and we compute $P_r$. Fix vertex $u \in V^{(o)}(G)$ with $d(u)>\Delta$ and edge $(v',v)$ with $v' \in V^{(r-1+o)}(G)$ and $v \in V^{(r+o)}(G)$, we have at most $O(m/\Delta)$ choices for $u$ and $m$ choices for $(v',v)$, hence at most $O(m^2/\Delta)$ total choices. Using the inductive assumption, for each pair $u,v'$ and any function $g : \{1,...,r-1\} \to \{\text{Low},\text{High}\}$ we have the value of $M_{r-1,\Delta,g}(u,v')$. Every $P_r$ in $G$ starting at layer $o$ must also contain a  $P_{r-1}$ starting a layer $o$. Hence, we can then set the function $f :\{1,...,r\} \to \{\text{Low},\text{High}\}$ that agrees with $g$ in each value and has $f(r) = \text{Low}$ if $d(v) \leq \Delta$ and $f(r) = \text{High}$ otherwise, and add the product of $M_{r-1,o,\Delta,g}(u,v')$  and the weight of $(u,v)$ to the entry $(u,v)$ in the hashmap $M_{r,o,\Delta,f}$. We only need constant time per pair $u,(v',v)$ and hence this step will take $\tilde{O}(m^2/\Delta)$ total time.
		
		The third quantity is equivalent to the second, just reversing the position of the high degree vertex, and it can be shown in a similar way selecting edges connecting colors $o$ and $o+1$ and vertices from $V^{r+o}$.
	\end{proof}
	
	We can now prove the combinatorial algorithm:
	
	\begin{lemma} \label{lem:wsub_cycles_comb}
		Let $G$ be a colored weighted graph with $m$ edges. For all $k>3$, there is a combinatorial algorithm that computes $\WSub{G}{\cycle{k}}$ in time $\tilde{O}(m^{2-1/\lceil k/2 \rceil })$.
	\end{lemma}
	\begin{proof}
		Let $\Delta = m^{1/\lceil k/2 \rceil}$. For every function $g ; \{1,...,k\} \to \{\text{Low}, \text{High}\}$ we compute the total weight $W_g$ of colorful $k$-cycles $\{v_1,..,v_k\}$ in $G$ such that the degree of the $v_i \leq \Delta$ if $g(i)=\text{Low}$ and $v_i > \Delta$ otherwise. We have that $\WSub{G}{\cycle{k}} = \sum_g W_g$. We distinguish two cases:
		\begin{itemize}
			\item $g(i) = \text{Low}$ for all $i$. In this case we can split the $k$-cycle into two paths $P_{\lfloor k/2 \rfloor}$ from layer $0$ to layer $\lfloor k/2 \rfloor$ and a $P_{\lceil k/2 \rceil}$ from layer $\lfloor k/2 \rfloor$ to layer $0$. We can then express $W_g$ as follows:
			\[
				W_g = \sum_{\substack {u \in V^{(0)} : d(u) \leq \Delta \\ v \in V^{(\lfloor k/2 \rfloor)}: d(v) \leq \Delta}} N_{\lfloor k/2 \rfloor,0,\Delta}(u,v) \cdot N_{\lceil k/2 \rceil,\lfloor k/2 \rfloor,\Delta}(v,u)
			\]
			Using \Lem{aux_cycle} we can compute $N_{\lfloor k/2 \rfloor,0,\Delta}(u,v)$ and $N_{\lceil k/2 \rceil,\lfloor k/2 \rfloor,\Delta}(v,u)$ for every pair $u,v$ in total time $\tilde{O}(m\Delta^{\lceil k/2 \rceil - 1})$, and we will have at most $O(m\Delta^{\lceil k/2 \rceil - 1})$ non-zero entries, thus the sum can be computed in time $\tilde{O}(m\Delta^{\lceil k/2 \rceil - 1}) = \tilde{O}(m^{2-1/\lceil k/2 \rceil}) $.
			
			\item There is at least an entry $i$ for which $g(i) = \text{High}$. In this case we split the $k$-cycle into two paths $P_{\lfloor k/2 \rfloor}$ from layer $i$ to layer $i+\lfloor k/2 \rfloor$ and a $P_{\lceil k/2 \rceil}$ from layer $i+\lfloor k/2 \rfloor$ to layer $i$. We set $f$ as the restriction of $g$ to the values $\{i+1,...,i+\lfloor k/2 \rfloor\}$ and $f'$ as the restriction of $g$ to the values  $\{i+\lfloor k/2 \rfloor\,...,i-1\}$ (modulo $k$). We can express $W_g$ as follows:
			\[
				W_g = \sum_{\substack {u \in V^{(i)} : d(u) > \Delta \\ v \in V^{(\lfloor k/2 \rfloor+i)}}} M_{\lfloor k/2 \rfloor,i,\Delta, f}(u,v) \cdot M'_{\lceil k/2 \rceil,\lfloor k/2 \rfloor+i,\Delta, f'}(v,u)
			\]
			Using \Lem{aux_cycle} we can compute $ M_{\lfloor k/2 \rfloor,i,\Delta, f}(u,v)$ and $M'_{\lceil k/2 \rceil,\lfloor k/2 \rfloor+i,\Delta, f'}(v,u)$ for all pairs $u,v$ in total time $\tilde{O}(m^2/\Delta)$. The number of pairs is bounded by $O(m^2/\Delta)$ as there are at most $O(m/\Delta)$ choices for $u$. Hence we can compute $W_g$ in time  $\tilde{O}(m^2/\Delta) = \tilde{O}(m^{2-1/\lceil k/2 \rceil})$.
		\end{itemize}
	\end{proof}	
	
	\subsection{Matrix multiplication algorithm for cycles}
	
		We prove now the matrix multiplication algorithm. The following lemma together with \Lem{wsub_cycles_comb} completes the proof of \Lem{wsub_cycles}. The structure of our proof follows closely Section $5$ in \cite{GiLeSh+23}. Note that the indexes over $\{0,...,k-1\}$ will be taken as modulo $k$.
		
		\begin{lemma} \label{lem:wsub_cycles_mm}
		Let $G$ be a colored weighted graph with $m$ edges. For all $k>3$, there is an algorithm that computes $\WSub{G}{\cycle{k}}$ in time $\tilde{O}(m^{c_k})$.
		\end{lemma}
		\begin{proof}
			Let the colors of $G$ go from $0$ to $k-1$. For each color $i \in [0,k-1]$, let $V^{(i)}$ be the set of vertices of $G$ with color $i$. We partition each of the sets into $\log m$ degree classes. For all $1 \leq j < \log{n}$ we define the set $W^{(i)}_j$ as follows:
			\[
				W^{(i)}_j = \{v \in V^{(i)} | 2^j \leq d(v) \leq 2^{j+1}\}
			\]
			Every such set will have at most $O(m/2^j)$ vertices. We classify the colorful $k$-cycles in $G$ into $O(\log^k{n})$ distinct classes, according to the degree of the vertices in the cycles. 
			
			We fix a tuple of degree classes $f = (f_0,...,f_{k-1})$ and show how to compute the weighted sum of cycles $\{u_0,u_1,\dots,u_{k-1}\}$ satisfying $u_i \in W^{(i)}_{f_i}$ for all $i$. Repeating this process for all classes will only add a poly-logarithmic term.
			
			Let $A^f_{i}$ be the $|W^{(i)}_{f_i}| \times |W^{(i+1)}_{f_{i+1}}|$ weighted adjacency matrix of $G$ restricted to the edges from vertices in $W^{(i)}_{f_i}$ to vertices in $W^{(i+1)}_{f_{i+1}}$. $A^f_{i}$ can only have $O(m)$ non-zero entries, and hence we can construct a sparse representation of it in $O(m)$ time. For every $p,q \in \{0,\dots,k-1\}$ we set:
			\[
				B^f_{p,q} = A^f_p \times A^f_{p+1} \times \dots \times A^f_{q-1}
			\]
			Note that the trace of the matrix $B^f_{0,k}$ will be equal to the sum of the product of weights of the $k$-cycles in $G$ satisfying the degree classes $f$. We can compute trace($B^f_{0,k}$) by selecting two points $p,q \in \{0,\dots,k-1\}$ and computing instead trace($B^f_{p,q}B^f_{q,p}$) = trace($B^f_{0,k}$).
			
			For all $i \in \{0,k-1\}$, we set $d_i = f_i/\log{m}$. Hence we can represent the tuple of array classes as $d = \{d_0,...,d_{k-1}\}$. The value of $d_i$ ranges from $0$ to $1$ and the sets $W^{(i)}_{f_i}$ will have at most $O(m^{1-d_i})$ vertices. For every $B^f_{p,q}$, let $P^f_{p,q}$ be the minimum such that we can compute $B^f_{p,q}$ in $O(m^{P^f_{p,q}})$ time (and hence $B^f_{p,q}$ will have at most $O(m^{P^f_{p,q}})$ non-zero entries). We can compute $B^f_{p,q}$ in three different ways:
			
			\begin{enumerate}
				\item First we compute a sparse representation of $B^f_{p,q-1}$. For every entry $(u,v) \in W^{(p)}_{f_p} \times W^{(q-1)}_{f_{q-1}}$ of $B^f_{p,q-1}$, we look at the neighbors of $v$ in $W^{(q)}_{f_{q}}$, for every such neighbor $v'$ we update $B^f_{p,q}(u,v') +\!\!= B^f_{p,q-1}(u,v) \cdot w((v,v'))$. Computing $B^f_{p,q-1}$ will take $O(m^{P^f_{p,q-1}})$ time and it will have as much non-zero entries. Every vertex $v$ will have at most $O(m^{d_{q-1}})$ neighbors. Hence, $B^f_{p,q}$ can be computed in time $O(m^{P^f_{p,q-1} + d_{q-1}})$.
				\item Similar to the previous method but reversed: We compute $B^f_{p+1,q}$ and then for every entry $(u,v)$ of $B^f_{p+1,q}$ we look at the neighbors $u'$ of $u$ in $W^{(p)}_{f_{p}}$ and update $B^f_{p,q}(u',v) +\!\!= B^f_{p+1,q}(u,v) \cdot w((u',u))$. We can then compute $B^f_{p,q}$ in time $O(m^{P^f_{p+1,q} + d_{p+1}})$.
				\item For some $p < r < q$, we compute $B^f_{p,r}$ and $B^f_{r,q}$ and then compute their product to obtain $B^f_{p,q}$. We need $O(m^{P^f_{p,r}} + m^{1-d_p} + m^{1-d_r})$ time to compute the matrix representation of $B^f_{p,r}$, similarly we need $O(m^{P^f_{r,q}} + m^{1-d_r} + m^{1-d_q})$ to obtain the matrix representation of $B^f_{r,q}$. We need additional $O(m^{M(1-d_p,1-d_r,1-d_q)})$ time to compute the product, where $M(a,b,c)$ is the minimum value of $g$ such that we can multiply matrices with dimensions $m^a \times m^b$ and $m^b \times m^c$ in time $O(m^g)$.
			\end{enumerate}
			
			We can now bound the value of $P^f_{p,q}$:
			\begin{itemize}
				\item From the first method we have $P^f_{p,q} \leq P^f_{p,q-1} + d_{q-1}$.
				\item From the second method we have $P^f_{p,q} \leq P^f_{p+1,q} + d_{p+1}$.
				\item From the third method we have: 
				\[P^f_{p,q} \leq \min_{p<r<q}\max \{P^f_{p,r},P^f_{r,q}, M(1-d_p,1-d_r,1-d_q)\}\]
			\end{itemize}
			
			Any $B^f_{p,p+1}$ can be constructed in time $O(m)$, hence $P^f_{p,p+1}=1$. Therefore, we can define $P^f_{p,q}$ as follows:
			
			\begin{equation} \label{eq:cycle1}
			\begin{split}
				& P^f_{p,p+1}=1 \\
				& P^f_{p,q} = \min\left\{ P^f_{p,q-1} + d_{q-1}, P^f_{p+1,q} + d_{p+1}, \min_{p<r<q}\max \{P^f_{p,r},P^f_{r,q}, M(1-d_p,1-d_r,1-d_q)\}    \right\}
			\end{split}
			\end{equation}
			
			Given the sparse representations of $B^f_{p,q}$ and $B^f_{q,p}$ we can compute trace($B^f_{p,q}\times B^f_{q,p}$) in time equal to the number of non-zero entries in each of the matrices $O(m^{P^f_{p,q}}+ m^{P^f_{q,p}})$. We define $C_k(f)$ as the minimum $g$ such that we can compute trace($B^f_{0,k}$) in time $O(m^g)$, we will have:
			
			\begin{equation}\label{eq:cycle2}
				C_k(f) = \min_{0\leq p < q \leq k-1} \max{P^f_{p,q},P^f_{q,p}}
			\end{equation}
			
			Equations \Eqn{cycle1} and \Eqn{cycle2} are identical to the equations in \cite{GiLeSh+23} and \cite{DaVuWi19}. Hence we can define $c_k$ as in \cite{GiLeSh+23}, giving:
			\begin{equation}
				c_k = \max_f C_k(f)
			\end{equation} 
			
			Which means that $O(m^{c_k})$ is an upper bound for computing trace($B^f_{0,k-1}$), for all $f$. We can then iterate over all tuples of degree classes $f$ and compute $\WSub{G}{\cycle{k}}$ in time $\tilde{O}(m^{c_k})$.
			
		\end{proof}
	
	\section{Proof of \Lem{lowerbound}}
		
		\subsection{Even cycles}
			We first prove \Lem{lowerbound} for even cycles.
			\begin{definition} [Expanded Graph]
				Given a graph $G$ with $n$ vertices and $m$ edges we define $\expandG$ as the graph resulting of subdividing every edge in $G$ by a $2$-edge path. The resultant graph has $m+n$ vertices and $2m$ edges.
			\end{definition}
		
		\begin{figure}
			\centering
			\includegraphics[width=\textwidth*3/4]{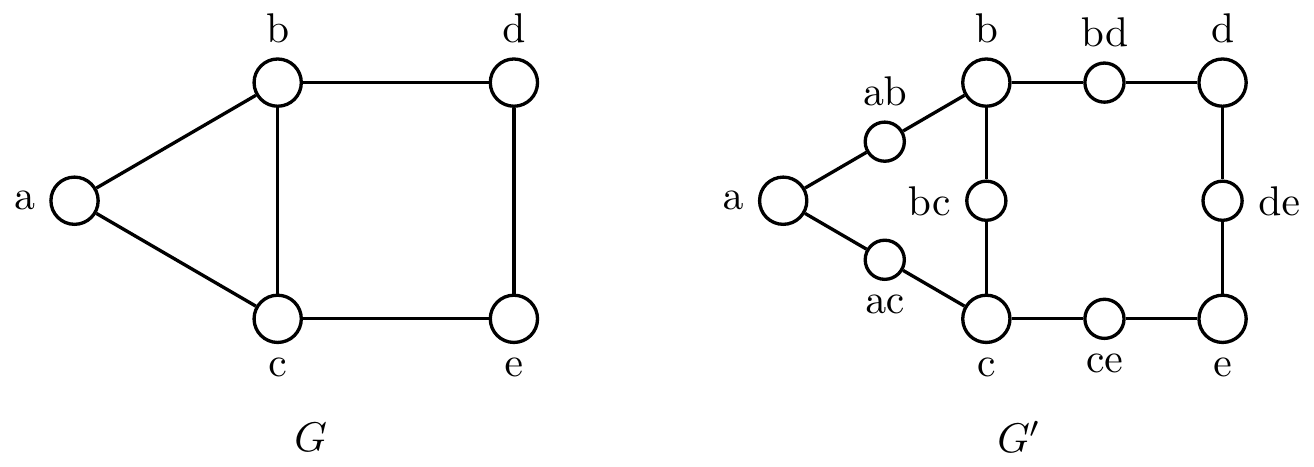}%
			\caption{An example of the construction of $\expandG$ from a graph $G$. The original graph had one $\cycle{3}$, one $\cycle{4}$ and one $\cycle{5}$. Those cycle become $\cycle{6},\cycle{8}$ and $\cycle{10}$ respectively.}
			\label{fig:subdivision}
		\end{figure}
		
		\Fig{subdivision} shows an example of the construction. We can show that $\expandG$ has constant degeneracy.
		
		\begin{claim} \label{lem:degen}
			The degeneracy of $\expandG$ is $2$.
		\end{claim}
		\begin{proof}
			Every connected subgraph with more than $1$ vertex must include one of the new vertices, the degree of these vertices is $2$ so the degeneracy can not be larger than $2$.
		\end{proof}
	
		We also show that there is a direct relation between the cycles in $G$ and in $\expandG$.
	
		\begin{lemma} \label{lem:equivalence}
			For any $k\geq 3$, the number of $\cycle{k}$ in $G$ is equal to the number of $\cycle{2k}$ in $\expandG$.
		\end{lemma}
		\begin{proof}
			Clearly every $k$-cycle in $G$ will become a $2k$-cycle in $\expandG$ as every edge is now a path of two edges. Now consider a cycle of length $2k$ in $\expandG$, we can can show that it corresponds to exactly one $k$-cycle in $G$. Let $v_1,...,v_{2k}$ be the vertices of the cycles, note that $\expandG$ is a bipartite graph with all the original vertices of $G$ in one side and the new intermediary vertices in the other, hence every alternate vertex will also be a vertex in $G$. Let $v_1,v_3,...,v_{2k-1}$ be those vertices, we can see that there must be an edge connecting $v_i$ with $v_{i+2}$ in $G$ for odd $i$ as there is a $2$-path connecting them in $\expandG$, hence the alternate vertices induce a $k$-cycle in $G$.
		\end{proof}
	
		We can now prove the following lemma.
		
		\begin{restatable}{lemma}{even} \label{thm:even}
			Let $c\geq 1$, if there is an $f(\degen)O(n^c)$ algorithm for counting $2k$-cycles, then there is a $O(m^c)$ algorithm for counting $k$-cycles. 
		\end{restatable}
		\begin{proof}
			Let $G$ be any graph with $n$ vertices, $m$ edges and degeneracy $\degen$, and assume we have a $f(\degen)O(n^c)$ algorithm for the $Sub_{\cycle{2k}}$ problem, we show that we can compute $\Sub{G}{\cycle{k}}$.
			
			First, compute the expanded graph $\expandG$ from $G$. This can be done in $O(m+n)$ time. Use the algorithm for counting the number of $2k$-cycles in $\expandG$, by \Lem{equivalence} we have that this amount will be equal to $\Sub{G}{\cycle{k}}$. The number of vertices in $G'$ is $O(m)$ and the degeneracy is $2$, hence the runtime will be $f(2)O(m^c) = O(m^c)$.
		\end{proof}
		
		The previous lemma, together with \Lem{uppercycle} gives \Cor{5-cycle}.

		\fivecycle*
		
		\subsection{Odd cycles}
		
		We now prove the lower bound for $7$ and $9$-cycles. We need to slightly modify our construction.
		
		\begin{definition} [Odd Expanded Graph]
			Given a graph $G$ with $n$ vertices and $m$ edges we define $\expandGOdd$ as the graph resulting of replacing every edge by a $2$-path and a $3$-path. The resultant graph has $3m+n$ vertices and $5m$ edges.
		\end{definition}
		
		\begin{figure}
			\centering
			\includegraphics[width=\textwidth*3/4]{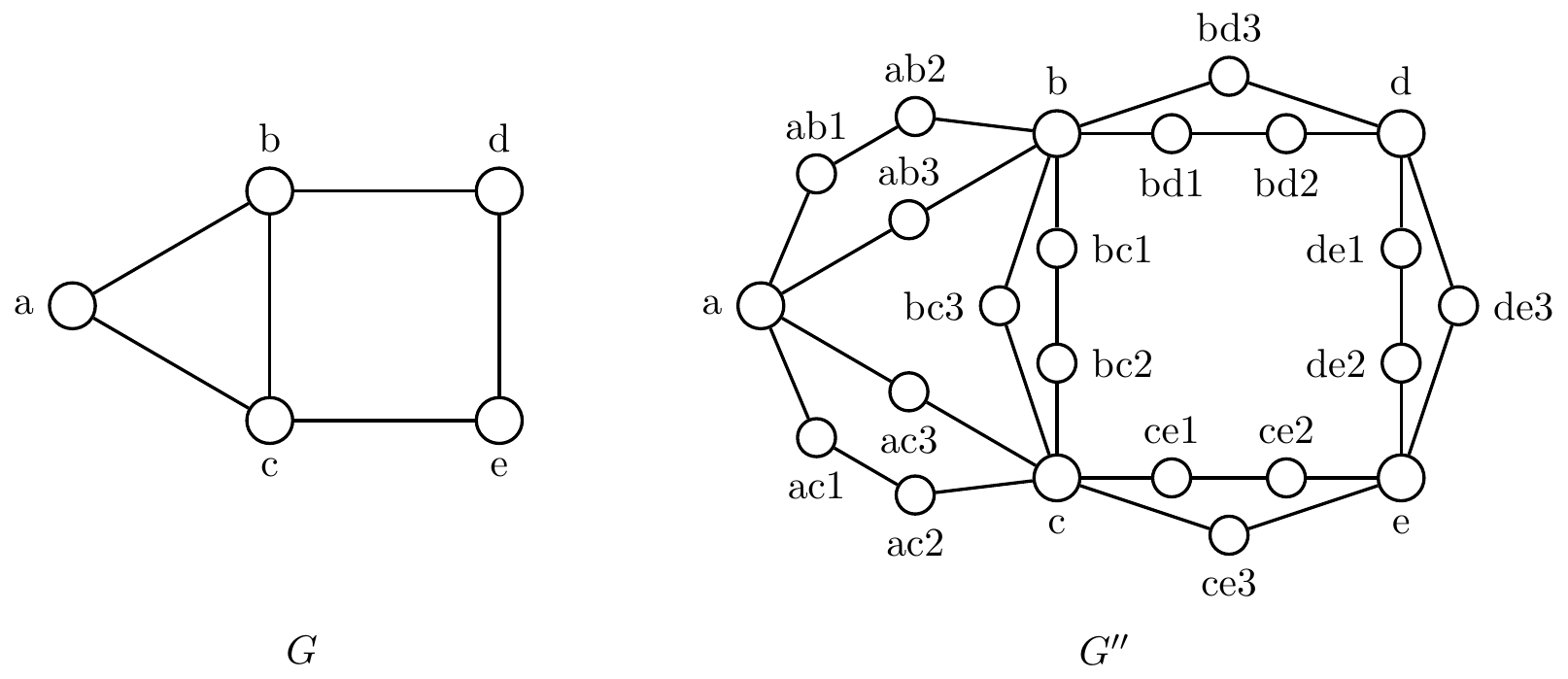}%
			\caption{An example of the construction of $\expandGOdd$ from a graph $G$. The original graph had one $\cycle{3}$ and one $\cycle{4}$. The expanded graph has $3$ $\cycle{7}$ and $5$ $\cycle{9}$.}
			\label{fig:subdivision2}
		\end{figure}
		
		\Fig{subdivision2} shows an example of the construction. Again, this graph will have a degeneracy of $2$. We can show the following:
		
		\begin{lemma} \label{lem:equivalence2}
			For any odd $k$:
			\begin{itemize}
				\item $\Sub{\expandGOdd}{\cycle{7}} = 3 \cdot \Sub{G}{\cycle{3}}$
				\item $\Sub{\expandGOdd}{\cycle{9}} = \Sub{G}{\cycle{3}} + 4 \cdot \Sub{G}{\cycle{4}} $
			\end{itemize}
		\end{lemma}
		\begin{proof}
			First consider any $\cycle{3}$ in $G$ given by vertices $v_1,v_2,v_3$, we can see that they will be part of exactly one $\cycle{9}$ in $\expandGOdd$: take the $3$-path connecting every pair of vertices from $v_1,v_2,v_3$, the result is a $\cycle{9}$. The same cycle will also corresponded to three different $\cycle{7}$ in $\expandGOdd$: For two of the pairs from $v_1,v_2,v_3$ take the $2$-path and take the $3$-path from the remaining pair, there are three possible pairs, giving three cycles.
			
			Now consider any $\cycle{4}$ in $G$, we can see that it will correspond with $4$ different $\cycle{9}$ in $\expandGOdd$: let the vertices $v_1,v_2,v_3,v_4$ be the four vertices of the cycle, for three of the pairs $(v_1,v_2),(v_2,v_3),(v_3,v_4),(v_4,v_1)$ take the $2$-path, and take the $3$-path for the remaining pair. The result is $4$ different $\cycle{9}$.
			
			Now inversely, consider a $\cycle{7}$ in $\expandGOdd$, note that it can have at most $3$ vertices from the original graph, and hence it can only correspond with $\cycle{3}$, giving the equality of the lemma. Similarly a $\cycle{9}$ in $\expandGOdd$ can have either $3$ or $4$ vertices from the original graph $G$, and hence it will correspond to either a $\cycle{3}$ or $\cycle{4}$.
		\end{proof}
		
		We can now show the lower bounds:
		
		\begin{lemma}
			\begin{itemize}
			\item If an $f(\degen)o(n^{d_3})$ algorithm for $\Sub{G}{\cycle{7}}$ exists then there exist a $o(m^{d_3})$ algorithm for $\Sub{G}{\cycle{3}}$.
			\item If an $f(\degen)o(n^{d_4})$ algorithm for $\Sub{G}{\cycle{9}}$ exists then there exist a $o(m^{d_4})$ algorithm for $\Sub{G}{\cycle{4}}$.
		\end{itemize}
		\end{lemma}
		\begin{proof}
			Assume an $f(\degen)o(n^{d_3})$ algorithm exists for $\Sub{G}{\cycle{7}}$. For any graph $G$ we can construct $\expandGOdd$ in $O(n+m)$ time. Then we can use the $f(\degen)o(n^{d_3})$ algorithm to compute the number of $\cycle{7}$ in $\expandGOdd$ from which we can get the number of $\cycle{3}$ in $G$ using \Lem{equivalence2}. $\expandGOdd$ has degeneracy of $2$ and $O(m)$ vertices, hence the algorithm will run in $o(m^{d_3})$.
			
			Similarly for $\cycle{9}$, assume a $f(\degen)o(n^{d_4})$ algorithm exists for $\Sub{G}{\cycle{9}}$. For any graph $G$ we can construct $\expandGOdd$ in $O(n+m)$ time, we can also use the existing $O(m^{d_3}) = o(m^{d_4})$ algorithm for computing the number of $\cycle{3}$ in $G$. We then use the assumed algorithm for $\cycle{9}$ in $\expandGOdd$. We can use \Lem{equivalence2} to compute the number $\cycle{4}$ in $G$ from the number of $\cycle{3}$ in $G$ and the number of $\cycle{9}$ in $\expandGOdd$. This will take $o(m^{d_4})$.
		\end{proof}

    \section{Proof of \Lem{hardness}} \label{sec:hardness}

	Consider the $10$-vertex graph $H_\triangle$ shown in \Fig{reduction_three}, it has one acyclic orientation $\vec{H}_\triangle$ that is \reducible{\hyperthree}, that can be seen by replacing every source in $\vec{H}_\triangle$  with a hyperedge.
	
	Let $G$ be any input graph with bounded degeneracy, we can show that a subquadratic algorithm for computing $\Sub{G}{H_\triangle}$ implies a subquadratic algorithm for computing $\Sub{\vec{G}}{\vec{H}_\triangle}$.
	
	\begin{lemma} \label{lem:aux_conj}
		If there is a $f(\degen)o(n^2)$ algorithm for computing $\Sub{G}{H_\triangle}$, then there is a $f(\degen)o(n^2)$ algorithm for computing $\Sub{\vec{G}}{\vec{H}_\triangle}$.
	\end{lemma}
	\begin{proof}
		\begin{equation}
			\Sub{G}{H_\triangle} = \sum_{\vec{H} \in \Sigma(H_\triangle)} \Sub{\vec{G}}{\vec{H}} = \sum_{\vec{H} \in \Sigma(H_\triangle)} \sum_{\vec{H'} \in \Spasm(\vec{H})}  f(\vec{H},\vec{H'})\Hom{\vec{G}}{\vec{H'}}
		\end{equation}
	The values of $f(\vec{H},\vec{H'})$ are non-zero for each pair $\vec{H},\vec{H'}$. We can hence rewrite $\Sub{G}{H_\triangle}$ as a linear combination of homomorphism counts. Consider the direct pattern $\vec{H}_\triangle$, note that because it has size $10$ it can not be in the spams of any other acyclic orientation of $H_\triangle$ besides its own. Hence it will have a non-zero coefficient of $f(\vec{H}_\triangle,\vec{H}_\triangle)$ in the linear combination.
	
	We can now use Lemma $3.1$ in \cite{GiLeSh+23} to transform our $f(\degen)o(n^2)$ algorithm for $\Sub{G}{H_\triangle}$ into a $f(\degen)o(n^2)$ algorithm for counting any of the homomorphisms with non-zero coefficient in the linear combination, hence we will have a $f(\degen)o(n^2)$ algorithm for $\Hom{\vec{G}}{\vec{H}_\triangle}$.
	
	Note that the rest of the patterns in $\Spasm(\vec{H}_\triangle)$ might not have non-zero coefficient in the linear combination, however, they will have at most $9$ vertices, and hence we can compute the number of homomorphisms for each of them in subquadratic time using \Thm{main}. We can then use the homomorphism counts of all the graphs in $\Spasm(\vec{H}_\triangle)$ to compute $\Sub{\vec{G}}{\vec{H}_\triangle}$ in total time $f(\degen)o(n^2)$.
	\end{proof}
	\begin{figure}
	\centering
	\includegraphics[width=\textwidth]{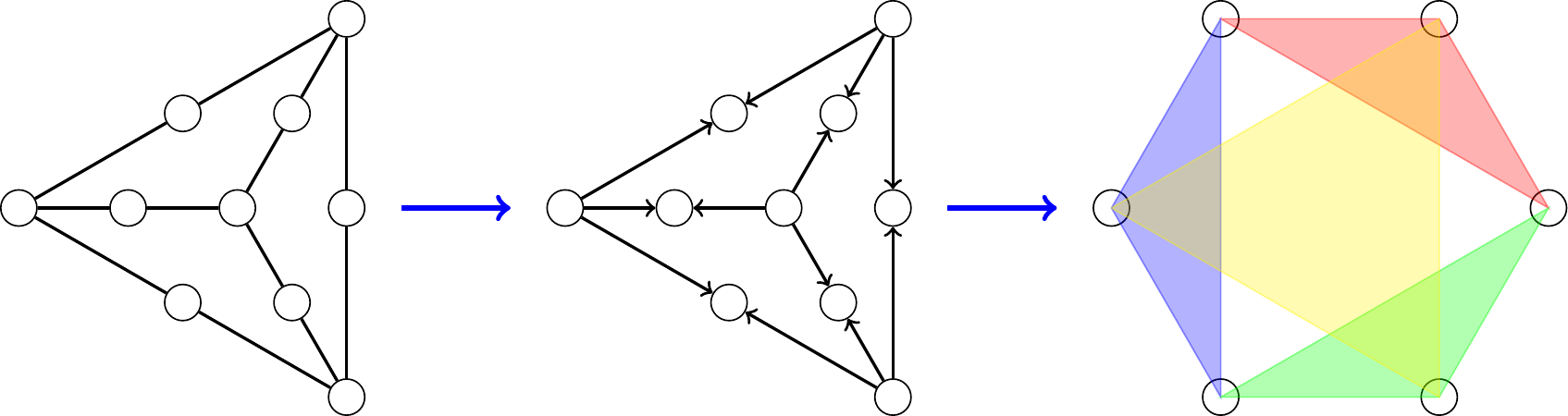}%
	\caption{The graph $H_\triangle$, its acyclic orientation $\vec{H}_\triangle$ that is \reducible{\hyperthree} and the hypergraph $\hyperthree$.}
	\label{fig:reduction_three}
	\end{figure}
		
	Now we can prove \Lem{hardness}.
	\hardness*
	\begin{proof}
		Assume there is a an algorithm that computes $\Sub{G}{H}$ in $f(\degen)o(n^2)$ time for all patterns with $10$ vertices. This implies that $\Sub{G}{H_\triangle}$ can be computed in $f(\degen)o(n^2)$ time. Using \Lem{aux_conj} we have that there will also be an $f(\degen)o(n^2)$ algorithm for computing $\Sub{\vec{G}}{\vec{H}_\triangle}$. We show how to compute $\Sub{G}{\hyperthree}$ using this algorithm.
		
		Let $G$ be a hypergraph with all hyperedges of arity $3$. Construct the graph $\vec{G'}$ by replacing every hyperedge $e=\{u_1,u_2,u_3\}$ by a new vertex $v_e$ and the directed edges $(v_e,u_1)$,$(v_e,u_2)$,$(v_e,u_3)$. This graph will have $O(m)$ edges and a degeneracy of $3$.
		
		Note that every copy of $\hyperthree$ in $G$ will now have become a copy of $\vec{H}_\triangle$ in $\vec{G'}$. Similarly every copy of $\vec{H}_\triangle$ in $\vec{G'}$ will correspond to a copy of $\hyperthree$ in $G$. Hence $\Sub{G}{\hyperthree} = \Sub{\vec{G'}}{\vec{H}_\triangle}$. We can then use the subquadratic algorithm that we assumed exists to compute $\Sub{G}{\hyperthree}$ in time $f(3)o(m^2) =o(m^2)$. But this contradicts \Conj{hyperthree}.
	\end{proof}

\end{document}